\documentclass[a4paper,USenglish,cleveref, autoref, thm-restate, numberwithinsect]{lipics-v2021}

\pdfoutput=1 %
\hideLIPIcs  %

\graphicspath{{./figures/}}%

\title{Transitions in Dynamic Point Labeling} %

\author{Thomas Depian}{Algorithms and Complexity Group, TU Wien, Vienna, Austria}{tdepian@ac.tuwien.ac.at}{https://orcid.org/0009-0003-7498-6271}{Funded by the Vienna Science and Technology Fund (WWTF) under grant ICT22-029.}%
\author{Guangping Li}{Algorithm Engineering Group, TU Dortmund, Dortmund, Germany}{guangping.li@tu-dortmund.de}{https://orcid.org/0000-0002-7966-076X}{Funded by the Austrian Science Fund (FWF) under grant P31119.}
\author{Martin N\"ollenburg}{Algorithms and Complexity Group, TU Wien, Vienna, Austria}{noellenburg@ac.tuwien.ac.at}{https://orcid.org/0000-0003-0454-3937}{}
\author{Jules~Wulms}{Algorithms cluster, TU Eindhoven, the Netherlands}{j.j.h.m.wulms@tue.nl}{https://orcid.org/0000-0002-9314-8260}{Funded partially by the Austrian Science Fund (FWF) under grant P31119 and partially by the Vienna Science and Technology Fund (WWTF) under grant ICT19-035.}

\authorrunning{T.\,Depian, G.\,Li, M.\,N\"ollenburg, J.\,Wulms} %

\Copyright{Thomas Depian, Guangping Li, Martin N\"ollenburg, and Jules~Wulms} %



\category{} %

\relatedversiondetails[cite={GISCIENCE}]{A preliminary version of this paper appeared in the proceedings of the 12th International Conference Geographic Information Science (GIScience'23)}{https://doi.org/10.4230/LIPIcs.GIScience.2023.2}
\relatedversiondetails[cite={CAGIS}]{The Version of Record of this manuscript has been published under open access and is freely available in Cartography and Geographic Information Science}{https://doi.org/10.1080/15230406.2025.2555426} %

\supplementdetails[cite={OSF}]{Data and Prototype used in the Case Study}{https://doi.org/10.17605/OSF.IO/HNSVU} %

\acknowledgements{The authors would like to thank GeoSphere Austria for providing access to their data. Furthermore, the authors acknowledge the use of Stadia Maps (\url{stadiamaps.com}) and OpenStreetMap (\url{openstreetmap.org/copyright}) in the prototype.}%

\nolinenumbers %

\usepackage{pifont} 				%
\usepackage{amsmath}
\usepackage{booktabs}
\usepackage{xcolor}
\usepackage{tikz}
\usepackage{url}
\usepackage[ruled,linesnumbered]{algorithm2e}
\SetKw{Break}{break}
\newcommand{\Labeling}[1]{\ensuremath{\mathcal{L}_{#1}}}
\newcommand{\GeneralLabeling}{\ensuremath{\mathcal{L}}}
\newcommand{\PlanarMonotoneMaxTwoSAT}{\textsc{Planar Monotone Max 2-Sat}\xspace}
\newcommand{\PlanarMaxTwoSAT}{\textsc{Planar Max 2-Sat}\xspace}
\newcommand{\PlanarThreeSAT}{\textsc{Planar 3-Sat}\xspace}
\newcommand{\NP}{\ensuremath{\textsf{NP}}}
\newcommand{\NPHard}{\mbox{\ensuremath{\textsf{NP}}-hard}\xspace}
\newcommand{\NPComplete}{\mbox{\ensuremath{\textsf{NP}}-complete}\xspace}

\newcommand{\TransitionStyle}[1]{\ensuremath{\mathit{#1}}}
\newcommand{\Transition}[1]{\ensuremath{\Labeling{1} \xrightarrow{\TransitionStyle{#1}} \Labeling{2}}}
\newcommand{\WeightedTransition}[2]{\ensuremath{\Labeling{1} \xrightarrow[#2]{\TransitionStyle{#1}} \Labeling{2}}}

\newcommand{\GoalOverlapas}{\ensuremath{\mathcal{G}_1}}
\newcommand{\GoalDuration}{\ensuremath{\mathcal{G}_2}}

\newcommand{\DirectedEdge}[2]{\ensuremath{#1 \to #2}}
\newcommand*\circled[1]{\tikz[baseline=(char.base)]{\node[shape=circle,draw,fill=white,font=\bfseries,inner sep=.75pt] (char) {\small #1};}}
\newtheorem{problem}{Problem}

\definecolor{labelcolorgreen}{rgb}{0.518 0.824 0}
\definecolor{labelcolorpurple}{rgb}{0.678 0.125 0.678}
\definecolor{labelcoloryellow}{rgb}{1 0.867 0}
\definecolor{labelcolorlime}{rgb}{0.808 0.875 0}
\newcommand\ProofLabelColor[1]{\textcolor{white!50!#1}{\rule{9 pt}{9 pt}}}

\newcommand{\DatasetAccidents}{\texttt{Accidents}\xspace}
\newcommand{\DatasetTwitter}{\texttt{Twitter}\xspace}
\newcommand{\DatasetWeather}{\texttt{Weather}\xspace}

\begin{document}

\maketitle

\begin{abstract}
The labeling of point features on a map is a well-studied topic. In a static setting, the goal is to find a non-overlapping label placement for (a subset of) point features. In a dynamic setting, the set of point features and their corresponding labels change, and the labeling has to adapt to such changes. To aid the user in tracking these changes, we can use morphs, here called \emph{transitions}, to indicate how a labeling changes. Such transitions have not gained much attention yet, and we investigate different types of transitions for labelings of points, most notably \emph{consecutive} transitions and \emph{simultaneous} transitions. 
We give (tight) upper bounds on the number of overlaps that can occur during these transitions.
When each label has a non-negative weight associated to it, and each overlap imposes a penalty proportional to the weight of the overlapping labels, we show that it is \NPComplete to decide whether the penalty during a simultaneous transition has weight at most~$k$.
Finally, %
we consider geotagged data on a map, by labeling points with rectangular or square labels. We developed a prototype implementation to evaluate different transition styles in practice, measuring both number of overlaps and transition duration.
\end{abstract}

\section{Introduction} %
Maps are ubiquitous in the modern world: from geographic to political maps, and from detailed road networks to schematized metro maps, maps are used on a daily basis. Advances in technology allow us to use digital maps on-the-fly and in a highly interactive fashion, by means of panning, zooming, and searching for map features. Besides changes induced by the user, maps can also change passively, for example automated panning during GPS routing, or changing points of interest when visualizing time-varying geospatial (point) data.

Important features on a map are often labeled. Examples of such features are areas (e.g., countries and mountain ranges), curves (e.g., roads and rivers), and most importantly points of interest (e.g., depending on the scale of the map, cities, landmarks, shops etc.).
The aforementioned interactions force map features and their corresponding labels to change, by appearing, disappearing, or changing position. Instead of swapping between the map before and after such changes, we can use morphs, here called \emph{transitions}, to allow the user to more easily follow changes in map features and labelings.

\begin{figure}
\centering
\includegraphics[width=\textwidth,page=2]{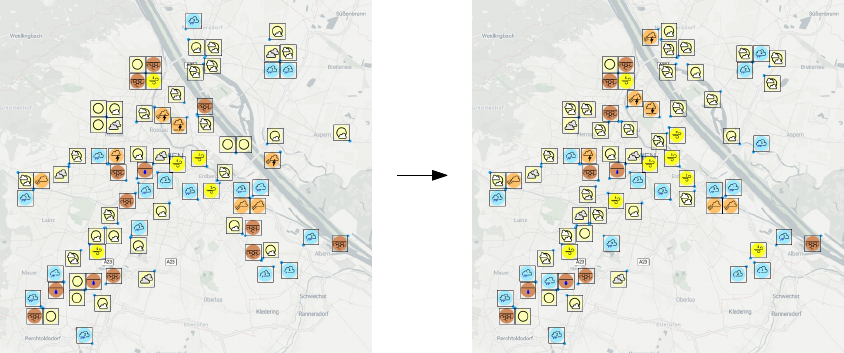}
\caption{Two labelings of the weather situation in Vienna at different timestamps: 23 labels appear, 26 disappear, and 11 change their position. The labelings are computed by our prototype that we describe in \Cref{sec:case-study}. Label-Icons: \textcopyright\ GeoSphere Austria~\cite{GeoSphereAustria2024}. Note that a full visual scan of the individual labels is necessary to identify all changes~\cite{Rensink.1997}.}
\label{fig:instant_change}
\end{figure}

\cref{fig:instant_change} shows why such transitions are important: even for two very similar point labelings, a lot of mental effort can be required to identify the differences.
This phenomenon has already been addressed in the literature:
User studies have shown that humans have difficulties in identifying the changes between two images, in particular if they occur in regions of the image that users deem irrelevant~\cite{Rensink.1997}.
To increase the perception of changes, Rensink, O'Regan, and Clark~\cite{Rensink.1997} suggest that the attention of a user should be drawn to the objects that change.
We see animated transitions as a way of interpolating between two labelings that, in addition, aids in highlighting the changes in the labeling. Furthermore, we believe that for applications that display user-generated content on a map, it is particularly important to give the user an indication of the changes in the labeling.
After all, for such data the labeling can change at any time, even without explicit interaction with the map itself.
Already Fish, Goldsberry, and Battersby~\cite{FGB.CBA.2011} noted that it is not only important that users recognize changes on dynamic maps, but they should also be able to follow and understand these changes.
Examples for such dynamic maps are numerous and range from mostly static map content, such as historic reports of traffic accidents~\cite{SFL2024}, to highly dynamic content, such as the current location of rental scooters~\cite{OBrien2024} or incidents reported to city services~\cite{StadtWien2024}.
In particular, we would like to highlight the service \emph{wettermelden.at}~\cite{GeoSphereAustria2024} offered by \emph{GeoSphere Austria}, the Austrian institute for Geology, Geophysics, Climatology, and Meteorology.
It provides a platform where users can report their local weather condition and view all reports on a map.
In addition, users can view the reports of the last twenty-four hours as a video-like animation showing a map of Austria with the reports at different (predefined) timestamps.
However, this animation lacks a transition between individual labelings, which makes it particularly difficult to track changes over time.
Therefore, we see this application as an ideal real-world example of where transitions could be implemented in practice.

\begin{figure}
	\centering
	\includegraphics[page=4]{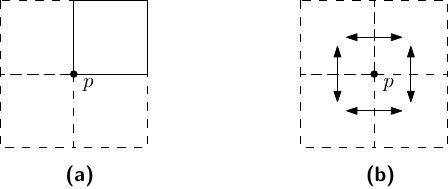}
	\caption{Visualization of the fixed and slider position model. \textbf{\textsf{(a)}} The four candidate positions for label~$l$ of point~$p$, with $l$ placed in the top-right position. \textbf{\textsf{(b)}} Labels continuously move between candidate positions using the sliding-position model.}
	\label{fig:position-model}
\end{figure}

In this paper we study transitions between labeled maps that show point features~$P$ and their labels~$L$, where each point $p_i\in P$ has a label $l_i\in L$ associated to it. Labels are axis-aligned rectangles in the classical and frequently used four-position model, that is, each point~$p_i$ has four possible candidate positions to place label~$l_i$~\cite{Formann.1991} (see \cref{fig:position-model}a).
Note that the considered candidate positions are the most preferable ones according to Yoeli~\cite{Yoeli.1972}.	However, our proposed transitions can readily be generalized to other fixed and sliding position models.
While labels are often modeled as arbitrary (axis-aligned) rectangles, we use squares with side length $\sigma = 1$ for simplicity of the exposition.
Furthermore, labels can be modeled as axis-aligned squares if they should display a single icon as in~wettermelden.at~\cite{GeoSphereAustria2024} or \cref{fig:instant_change}.
In Appendix~A, we show that our results on the number of possible label overlaps also extend to arbitrary-sized rectangles by multiplying the bound by a constant factor that depends on the ratio between the size of the rectangles. A labeling~$\GeneralLabeling \subseteq L$ of $P$ consists of a set of pairwise non-overlapping labels, and can be drawn on a map conflict-free, by drawing only the labels in \GeneralLabeling\space with their associated points.
If the label $l \in L$ for a point $p \in P$ is not contained in \GeneralLabeling, we do not draw $p$ either.

For the study of transitions, we work in a dynamic setting, where points appear and disappear at different moments in time, and the set~$P$ changes only through additions and deletions: the data we consider later consists of geotagged events, for which we know only the location at the moment they occur, and hence data points do not move. Every time changes are made to $P$, a new overlap-free labeling must be computed, thus resulting in a change from labeling~\Labeling{1}, before the changes, to labeling~\Labeling{2}, afterwards. In this paper, we study different types of transitions from~\Labeling{1} to~\Labeling{2}, namely consecutive and simultaneous transitions. In the former, only one label moves at a time, while in the latter, multiple labels can move simultaneously. During both types of transitions, the individual labels are allowed to move in the sliding-position model, i.e., are only allowed to move horizontally and vertically around their point while continuously touching it~\cite{vanKreveld.1999} (see \cref{fig:position-model}b). 
Although the obtained movement paths might not be the shortest ones, we see two benefits in using the sliding-position model:
First, it is the natural generalization of the fixed-position model and it guarantees that throughout the movement we maintain a valid, but not necessarily overlap-free, labeling.
Second, the association between a point and its label is unambiguous throughout the whole movement.
In particular, a label never occludes its point, which is not guaranteed if labels move on the shortest path directly to their target position, i.e., perform a linear transition.
Although we view the above two aspects as clear advantages over linear transitions and as a valid way of defining movements, we acknowledge that, to the best of our knowledge, there is  no user study yet validating this assumption.
We regard such empirical validation as a natural next step in the study of transitions in point labeling, positioning our contribution as an initial algorithmic exploration of this topic. In particular, our aim is to find transitions that achieve optimization criteria, such as minimizing the number of overlaps during a transition, or minimizing the time required to perform a transition. To our knowledge, this is the first time transitions have been studied systematically in this way.

\subsection{Contributions}
First, we consider the limits of the proposed transition models in terms of worst-case behavior and computational complexity from a theoretical perspective. More precisely, in \cref{sec:consecutive,sec:rma}, we analyze the worst-case number of label overlaps during consecutive and simultaneous transitions, respectively. 
In \cref{sec:rma} we additionally consider instances where we associate weights to the labels (and to their overlaps) and prove that it is \NPHard to minimize the weight of overlaps if all labels in a transition move simultaneously.
However, the instances that we consider there to construct the worst-case examples are artificially crafted in order to reach the limits and do not necessarily resemble those that we have to deal with in an actual application.
Thus, we complement our theoretical results and consider transitions from a more applied point of view in \cref{sec:case-study}.
We present an implementation of our proposed transition styles as an interactive prototype and investigate in a detailed case study on three different types of dynamic real-world data how the transition styles perform on our desired optimization goals.
The experiments reveal a trade-off between the optimization criteria and suggest different transition styles depending on the most important goal.
Furthermore, the results also indicate that some transition styles studied in \cref{sec:consecutive} are mainly of theoretic interest.

\subsection{Related Work}
Automated map labeling is a well-researched topic within the geographic information science (GIS) and computational geometry communities. In recent years, the labeling of road networks~\cite{NiedermannN16}, island groups~\cite{GoethemKS16}, time-varying maps~\cite{Barth.2016,Krumpe18}, combining labeling with word clouds~\cite{BhoreG0NW21,bglnw-wapliwc-23}, and using human-in-the-loop approaches for labeling~\cite{Klute.2019} have been investigated. Previous literature has focused on (the complexity of) computing labelings in various static~\cite{Agarwal.1998,Formann.1991,vanKreveld.1999} and dynamic or kinetic settings~\cite{Bhore.2020,BuchinG14,BergG13}. 
Most of the existing work on map labeling has focused on label legibility, i.e., required that the labels are pairwise overlap-free, but disregarded potential overlaps with map features.	For fixed positions, overlaps between a label and an important map feature can be prevented by removing certain candidate positions for the label.
However, a considerable gap remains between the desired quality criteria and those that existing algorithms support.
Nevertheless, algorithms for these restricted settings provide the foundation for more advanced labeling algorithms that consider a multitude of objectives and constraints; see, for example, Rylov and Reimer~\cite{Rylov.2014} for a formalization of more general quality criteria.
In light of this, we require only that the labelings between which we compute the transition are overlap-free.
Therefore, label legibility will serve as the primary quality criterion in our study; incorporating additional criteria is a  promising direction for future research.

Our problem description resembles the ones of work mentioned above and also other earlier work on point labeling, but it also has subtle differences: For example, the optimization criteria we care for, minimizing overlaps and time required for labels to move, were already investigated by de Berg and Gerrits~\cite{BergG13}. They showed that there often is a clear trade-off between these criteria when dealing with moving labels. However, in their model, points are allowed to move (even during label movement), while our points are static, and change only through additions and deletions. Furthermore, in the \textsc{PSPACE}-hardness framework by Buchin and Gerrits~\cite{BuchinG14} points are often static and only labels move. Hence, their dynamic labeling instances are similar to transitions. Though, a distinct difference is that labels must be allowed to move back and forth in the dynamic labeling instances of the hardness reduction. Since we disallow detours in transitions (see~\Cref{sec:prob-desc}), this reduction is not easily transferred to our setting.
In the field of dynamic or interactive map labeling, controlling the position of the labels after a map interaction plays a central role for guaranteeing the labeling to remain consistent.
Been, Daiches, and Yap~\cite{Been.2006} were among those that pioneered this field of research by proposing four desiderata that a consistent dynamic labeling should achieve under zooming, panning and label selection operations.
They modeled each label as a cone to represent it across different zoom levels.
Computing a consistent labeling then corresponds to the problem of finding for each cone a range of zoom levels on which the label should be displayed.
This resulted in the \emph{active range optimization} problem, whose flavors have found considerable attention in the literature~\cite{Been.2010,Liao.2016,Zhang.2020,Gemsa.2011b}.
Prior to Been et al.~\cite{Been.2006},
Petzold, Gr{\"o}ger, and Pl{\"u}mer~\cite{Petzold.2003} used a conflict graph capable of tracking label overlaps across different map scales to label points, lines, and areas on interactive maps.
The versatility of their approach comes at the cost of not having a consistency guarantee for the computed labelings.
Apart from modeling concrete interactions with a map, the literature also contains approaches that work on an abstract activity interval that captures the points in time when a label can be displayed on a map~\cite{Gemsa.2013,Barth.2016}.
Furthermore, other work on dynamic map labeling either did not consider the desiderata by Been et al.~\cite{Poon.2005,Klute.2019}, or enforced for a label the same position independent of the visible map cutout, thus trivially satisfying it~\cite{Krumpe18,Gemsa.2016a}.
Note that by using the sliding position model to perform our movements, we will achieve this desideratum.
However, we only focus on the transition between two given labelings and neglect the non-trivial problem of computing two consistent labelings.
Our work is also closely related to the one by Schwartges, Haunert, Wolff, and Zwiebler~\cite{SchwartgesHWZ14}, who considered the problem of labeling a set of points on a dynamic map in the one-slider model.
However, their paper differs not only in terms of the position model used, but also in the underlying problem.
To that end, they focus on finding labelings for different map cutouts for a time-independent point set, which when combined should yield a consistent labeling of the map, but they do not consider the actual transition between the individual labelings of different cutouts.
Furthermore, moving polygonal labels on the plane to remove overlaps or increase their legibility has also been studied in the past~\cite{Garderen.2017,Gaertner.2024}.
Finally, our analysis of the number of overlaps in transitions draws multiple parallels with the analysis of topological stability, introduced in the framework for algorithm stability~\cite{MeulemansSVW18}. This framework provides various (mathematical) definitions of stability for algorithms on time-varying data: Intuitively, small changes in the input of an algorithm should lead to small changes in the output. One such stability definition, called topological stability, prescribes that the output changes continuously. The (topological) stability ratio of an algorithm then measures how close to optimum the stable output is: when an optimal solution undergoes a discrete change, a topologically stable output has to continuously morph through suboptimal solutions. Similarly, we analyze in this paper transitions with continuous movements in the sliding position model of various styles. We then analyze how close to overlap-free a labeling is during a transition by counting overlaps.

Our work is motivated by the lack of a systematic and algorithmic analysis of approaches to visualize the changes between two point labelings.
Such animations are widely used in modern maps, and understanding their effectiveness and applicability in GIS applications and beyond has gained attention in the literature; see also the overviews of Harrower~\cite{Har.CLA.2007} and Cybulski~\cite{Cyb.CAM.2024} or the survey by Harrower and Fabrikant~\cite{HF.RMA.2008}.
The justification for animations lies in the known limitation of humans in recognizing changes between static images and understanding the changes~\cite{Har.CLA.2007}.
However, their potential depends to a great extent on their implementation~\cite{Har.CLA.2007}, see also Fabrikant and Goldsberry~\cite{FG.Trp.2005}, and by now the literature proposes a set of design principles~\cite{MM.NWR.2003}.
The known limitation in perceiving differences mentioned above is often referred to as \emph{change blindness}~\cite{SR.Cbp.2005}; see also Stehle and Kitchin~\cite{SK.Rta.2019} for an overview of previous work on this topic.
Goldsberry and Battersby~\cite{GB.ICD.2009} suggested three levels of \emph{change detection}, or, symmetrically, change blindness: recognizing a change, detecting the type of change, and understanding the meaning of a change, see also Fish and Goldsberry, and Battersby~\cite{FGB.CBA.2011}.
Although these levels where introduced for choropleth maps, they also find applications outside their intended domain.
In particular, with our transitions, we aim at achieving the highest level of change detection, where the user is able to relate the individual labels in the two labelings.
Fish et al.~\cite{FGB.CBA.2011} assessed the influence of abrupt and smooth changes in choropleth maps on the change detection level.
They concluded that the design of transitions does matter, while favoring smooth over abrupt changes.
Krassanakis, Filippakopoulou, and Nakos~\cite{KFN.Dmp.2016} investigated the minimum duration of a change to be perceivable by humans.
They concluded that it takes around 400 milliseconds for participants to detect a moving point on a map.

In the next section, we define the problem that we consider in this paper, discuss the components of a transitions, and the goals they should achieve.
Although there are numerous ways to describe a transition and its desired properties, we believe that the model we study strikes a good compromise between simplicity, flexibility, and applicability.
In particular, some of the modeling choices we make are motivated by the findings discussed above.

\subsection{Problem Description}
\label{sec:prob-desc}
Let~$P$ be a finite point set of points in $\mathbb{R}^2$ and $L$ a set of labels, where each point $p_i\in P$ has a label $l_i\in L$ associated to it.
Given two (overlap-free) labelings \Labeling{1} and \Labeling{2} of $(P,L)$, i.e., two subsets of $L$, each containing only pairwise non-overlapping labels, we denote a transition between them with \Transition{}. Such a transition consists of changes of the following types.
\begin{description}
	\item[Additions] If only label $l_i$ of a feature point $p_i$ must be added, we denote this by \Transition{A_i}.
	\item[Removals] If only label $l_i$ of a feature point $p_i$ must be removed, we denote this by  \Transition{R_i}.
	\item[Movements] If only label $l_i$ of a feature point $p_i$ must change from its position in~\Labeling{1} to a new position in~\Labeling{2}, we denote this by \Transition{M_i}.
	Movements are unit speed and axis-aligned, in the sliding-position model.
	Note that a diagonal movement, as in \cref{fig:transition-goal-visualization}a (left), is composed of a horizontal and a vertical movement, and hence takes two units of time.
\end{description}
A label is \emph{stationary} if it remains unchanged during a transition.
Applying multiple transitions consecutively is indicated by chaining the corresponding transition symbols: \Transition{M_i M_j} denotes that label $l_i$ moves before label $l_j$.
Furthermore, \Transition{M} is a shorthand for applying all movement-transitions simultaneously. All these notions extend to additions and removals, using \TransitionStyle{A} and \TransitionStyle{R}, respectively, instead of \TransitionStyle{M}.
A transition has no effect if no point must be transformed with the respective transition, e.g., even if there are no additions, the transition \Transition{A} is still applicable; it simply does not modify the labeling.

We aim to identify types of transitions that try to achieve the following goals.

\begin{figure}
	\centering
	\includegraphics[page=6]{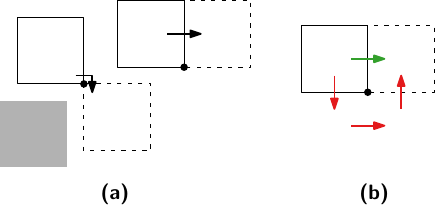}
	\caption{Transition goals studied in this paper. \textbf{\textsf{(a)}} \GoalOverlapas: Minimizing overlaps by moving around the gray stationary label. \textbf{\textsf{(b)}} \GoalDuration: Minimizing duration by using a single movement along the green arrow, instead of moving along the red arrows.}
	\label{fig:transition-goal-visualization}
\end{figure}

\begin{description}
	\item[\GoalOverlapas -- Minimize overlaps]
	While the two labelings are overlap-free, overlaps can occur during the transition from \Labeling{1} to \Labeling{2}. When too many overlaps happen at the same time, certain labels may (almost) completely disappear behind others during a transition, which defeats the purpose of the transition: allowing users to follow changes in the labeling.
	Thus, those overlaps should be avoided as much as possible, by, for instance, adjusting the movement direction of labels, as shown in \cref{fig:transition-goal-visualization}a. 
	After all, we see transitions as a way of interpolating between two overlap-free labelings, and, therefore, should aim at minimizing the overlaps during the interpolation.
	\item[\GoalDuration -- Minimize transition duration]
	Our main goal is to show a map in a (mostly) static state.
	However, we do not want to instantly swap \Labeling{1} for \Labeling{2}, since the user will have difficulties tracking all changes~\cite{Rensink.1997}. Though, a transition that takes too long can also cause users to lose attention~\cite{Nielsen.1993}.
	Hence, we want transitions that can be completed in a short amount of time.
	This can be achieved by disallowing detours, as in \cref{fig:transition-goal-visualization}b, 
	or by performing the changes simultaneously.
\end{description}

Note that ideally, one would also try to minimize the number of moving labels, as studies have shown that the amount of information humans can process is limited~\cite{Miller.1956}.
However, in this paper, we assume that we are given the two labelings \Labeling{1} and \Labeling{2}, which thus dictate the labels that have to move.
Therefore, we see the task of computing a \emph{stable} labeling \Labeling{2}, i.e., one in which only a few labels move, as an interesting research question in its own right.
Furthermore, we focus on optimizing the transition \emph{between} the labelings and assume for the remainder of this paper that the challenging task of computing a good labeling is solved separately.

Optimizing both goals \GoalOverlapas\ and \GoalDuration\ simultaneously is often impossible as there can be a trade-off: performing the transition as fast as possible to achieve \GoalDuration\space often leads to unnecessary overlaps, while preventing as many overlaps as possible to achieve \GoalOverlapas\space may require more time.
However, to work towards both \GoalOverlapas\space and \GoalDuration, we can perform all additions simultaneously, as well as all removals.
Furthermore, if we perform removals before movements, and movements before the additions, we create free space for the movements, to reduce the number of overlaps without wasting time. Let \TransitionStyle{X} be an arbitrary way of performing all movements required to change from \Labeling{1} to \Labeling{2} (consecutively or simultaneously), then we can observe the following.
\begin{observation}
	\label{obs:tg1-tg2}
	A transition of the form \Transition{RXA} aids in achieving both \GoalOverlapas\space and \GoalDuration.
\end{observation}
By separating additions, removals, and movements in transitions we also group together related labels, in our case those with the same change pattern, as suggested by Harrower~\cite{Har.CLA.2007}.

We introduce two overarching \emph{transition styles} in this paper: \emph{consecutive transitions} and \emph{simultaneous transitions}. Each such transition style is a variant of the style \TransitionStyle{RXA}, as prescribed by \cref{obs:tg1-tg2}, and fills in the movement described by \TransitionStyle{X} in a unique way. For a consecutive transition the movement \TransitionStyle{X} consists of a sequence of individual label movements, whereas for a simultaneous transition we have \TransitionStyle{X=M}. These transition styles each incur different transition durations. Since we expect a trade-off between \GoalOverlapas\ and \GoalDuration, we specifically analyze the number of overlaps during transitions of the two styles.

\section{Consecutive Transitions}
\label{sec:consecutive}
We begin our investigation by considering transition styles where one label at a time moves.
Although, from a practical perspective, this type of movement may be less favorable, especially if many labels need to be moved, we consider consecutive movements as the most fundamental movement pattern.
Accordingly, analyzing consecutive transitions serves as a foundational step towards identifying and understanding more sophisticated movement patterns, some of which are discussed in later sections of this paper while others are left for future work.
		
For consecutive transitions, we first upper-bound in \cref{sec:naive} the number of overlaps for an arbitrary order of the moving labels.
In \cref{sec:dag}, we analyze a more elaborate transition style that takes dependencies among moving labels into account to prevent certain overlaps.
\subsection{Naive Transitions}
\label{sec:naive}
We start with evaluating the potential overlaps for a single label performing its movement.
\cref{fig:naive-eight-overlaps}a shows how only a single stationary square label can interfere with the moving label.

\begin{lemma}
	\label{lem:one-point}
	In \Transition{R M_i A}, where only label $l_i$ moves, at most one overlap can occur. 
\end{lemma}
\begin{proof}
	As we perform removals before the movement and additions afterwards, we can guarantee that the start and end positions of label $l_i$ are overlap-free. Thus any overlap can occur only during diagonal movement of $l_i$, when $l_i$ moves from one candidate position in \Labeling{1}, to a non-adjacent candidate position in \Labeling{2}. Assume without loss of generality that $l_i$ traverses the lower-left label position, when moving from top-left to bottom-right. Only a single other (stationary) label $l_j$ can be positioned such that both \Labeling{1} and \Labeling{2} are overlap-free and the label overlaps with the area traversed by $l_i$ (see \cref{fig:naive-eight-overlaps}a). 
	Any additional label overlapping the traversed area, without overlapping $l_j$, would overlap the start or end position of $l_i$.
\end{proof}

Next we consider an arbitrary order of all $n$ moving labels in a transition. 
We define a conflict graph, which has a vertex for each moving label, and an edge between overlapping labels. With a packing argument we locally bound the degree of each of the $n$ moving labels to fourteen by considering the start, intermediate, and end position of such a label (these overlaps are achieved in \cref{fig:naive-eight-overlaps}b).
By the handshaking lemma, which states that in each graph the sum of vertex degrees equals twice the number of edges, this results in at most $7n$ overlaps.

\begin{figure}
	\centering
	\includegraphics[page=5]{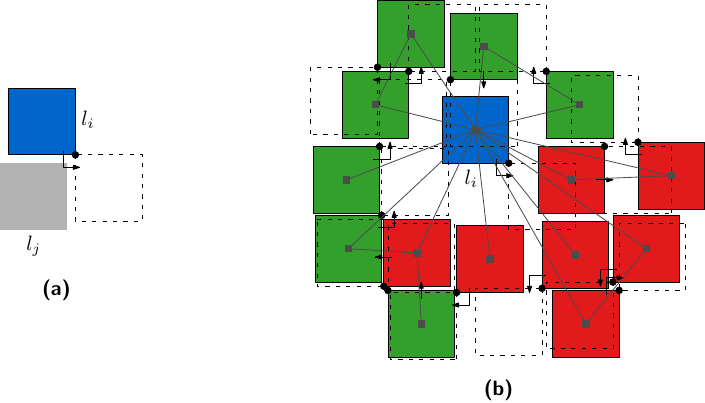}
	\caption{Constructions that achieve the claimed upper bound on the number of overlaps that we show in \textbf{\textsf{(a)}} \cref{lem:one-point} and \textbf{\textsf{(b)}} \cref{lem:consecutive-movements}. \textbf{\textsf{(a)}} Since all labels are squares with side length~$\sigma$, the moving blue label~$l_i$ can overlap only a single gray stationary label~$l_j$. \textbf{\textsf{(b)}} The blue label~$l_i$ overlaps fourteen other labels during the movement transitions. The green labels move before $l_i$, red labels move after $l_i$. We show in gray one possible conflict graph for this setting.}
	\label{fig:naive-eight-overlaps}
\end{figure}

\begin{lemma}
	\label{lem:consecutive-movements}
	In \Transition{R M_1 \dots M_n A} at most $7n$ overlaps can occur.
\end{lemma}
\begin{proof}
	When moving labels consecutively, not only the traversed area of a label $l_i$ can be occupied, but also its start or end position can be occupied by another label~$l_j$ that moves earlier or later, respectively.
	Such a situation produces one overlap per $l_j$, and we will show that at most fourteen such labels can exist. 
	See \cref{fig:naive-eight-overlaps}b 
	for an example of the construction that achieves this number of overlaps. However, each overlap involves two labels in this case, and hence should not be counted double. To prevent this, we build a conflict graph, where each vertex represents a label and each edge between two vertices represents an overlap between the respective labels. If we can bound the degree of each vertex by a constant~$c$, in our case $c=14$, then the total number of overlaps for $n$ moving labels is at most $\frac{c}{2}n = 7n$, by the handshaking lemma.
	Let us consider the case in which label $l_i$ performs a diagonal movement.
	We later argue why this ultimately leads to the most overlaps for $l_i$.
	Assume w.l.o.g.~that label $l_i$ performs a diagonal movement from the top-left to the bottom-right position, moving through the bottom-left position. Let $L'\subseteq L$ be the set of labels overlapping $l_i$ during movement and see \cref{fig:naive-overlaps-split,fig:naive-overlaps-split-2} for our construction.
	
	The labels in~$L'$ which overlap the start position of $l_i$ can start overlapping only during the transition, as otherwise \Labeling{1} is not overlap-free. Thus the start position of $l_i$ can overlap only with end positions (and intermediate positions) of labels in~$L'$. Symmetrically, labels in~$L'$ overlapping the end position of $l_i$ cannot overlap their assigned position at the end of the transition, since \Labeling{2} would then not be overlap-free. Hence the end position of~$l_i$ can overlap with only start (and intermediate) positions of labels in~$L'$.
	
	Since the labels in~$L'$ can also move themselves, we can allow some of the labels in $L'$ to overlap while keeping \Labeling{1} and \Labeling{2} overlap-free. For example, we can overlap the start position of one label in~$L'$ with the end position of another label in~$L'$. This overlap cannot happen on the start and end positions of~$l_i$, as these positions of~$l_i$ overlap only with end and start positions of labels in $L'$, respectively, when \Labeling{1} and \Labeling{2} are overlap free. However, we can overlap the start and end positions of two labels in~$L'$, if those positions occupy the intermediate position of~$l_i$ in the bottom-left corner; \Labeling{1} and \Labeling{2} can then still be overlap-free.
	
	\begin{figure}
		\centering
		\includegraphics{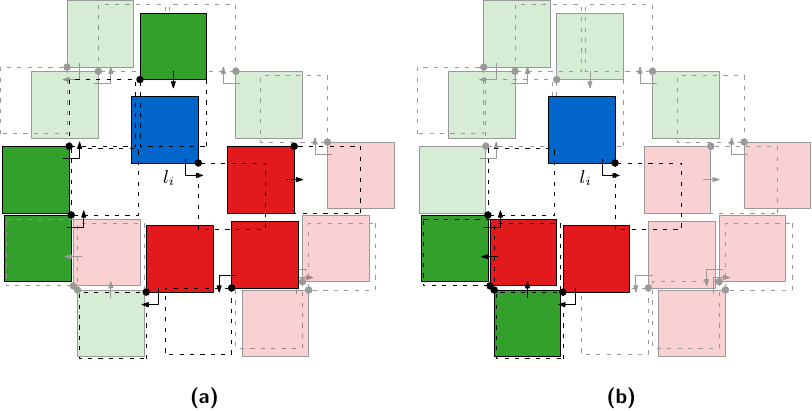}
		\caption{
			Parts of the constructed instance from \cref{lem:consecutive-movements}.
			The green labels move before~$l_i$, red labels move after~$l_i$.
			All labels form an instance that achieves the proved upper bound. \textbf{\textsf{(a)}} Labels overlapping with the blue label~$l_i$ at their start/end position. \textbf{\textsf{(b)}} Labels overlap~$l_i$ at its intermediate position.
		}
		\label{fig:naive-overlaps-split}
	\end{figure}
	
	First consider the start position of~$l_i$. Labels in $L'$ intersecting this position should necessarily move before $l_i$, as otherwise $l_i$ has moved away from its start position. The start position can overlap with at most three pairwise overlap-free end positions of labels in~$L'$ (see \cref{fig:naive-overlaps-split}a), as the end position of a fourth label would necessarily also overlap the end position of~$l_i$. Furthermore, intermediate positions of diagonally moving labels in~$L'$ can also overlap with the start position of~$l_i$. If we want such labels to overlap only the start position of~$l_i$, and not the end or intermediate position of~$l_i$, then there are at most four such labels in~$L'$ (see \cref{fig:naive-overlaps-split-2}a): two intersecting the top left corner of~$l_i$, one moving diagonally to the top-right, and another moving diagonally to the bottom-left, and two labels moving symmetrically at the top right corner of~$l_i$.
	
	Second, consider the end position of~$l_i$. Labels in $L'$ intersecting this position necessarily move after $l_i$, to ensure $l_i$ has reached its end position. For labels in~$L'$ starting on the end position of~$l_i$ we can repeat the whole argument for the start position of~$l_i$, but exchange ``start position'' by ``end position'', and vice versa (see again \cref{fig:naive-overlaps-split}a). For labels in~$L'$ whose intermediate position overlaps with only the end position of~$l_i$, we know that they must overlap in the top-right and bottom-right corner. Again for each corner there are at most two labels moving in opposite diagonal directions along the corners (see \cref{fig:naive-overlaps-split-2}b).
	
	\begin{figure}
		\centering
		\includegraphics{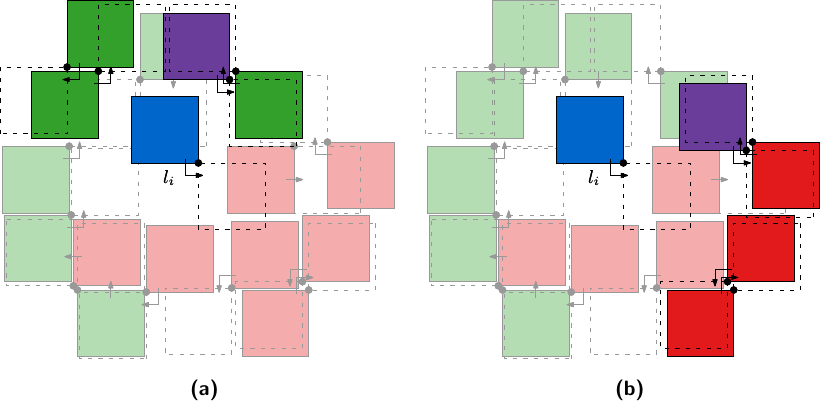}
		\caption{
			Labels whose intermediate positions overlap with the blue label~$l_i$ at its start and end position.
			The green labels move before~$l_i$, red labels move after~$l_i$.
			Only purple labels overlap with other labels.
			All labels form an instance that achieves the proved upper bound from \cref{lem:consecutive-movements}. \textbf{\textsf{(a)}} Labels overlapping with~$l_i$ at its start position. \textbf{\textsf{(b)}} Labels overlapping with~$l_i$ at its end position.}
		\label{fig:naive-overlaps-split-2}
	\end{figure}
	
	Third, consider the intermediate position of~$l_i$ and see \cref{fig:naive-overlaps-split}b for the following argumentation. Labels whose intermediate position overlaps with the intermediate position of~$l_i$ are not interesting: such labels will never overlap, since movement is consecutive. The intermediate position of $l$ can hence overlap with at most four labels in~$L'$, two labels extending into the start and end positions of~$l_i$, respectively, and two labels, with overlapping start and end positions, in the corner opposite of the corner where the point $p_i$ of label~$l_i$ is located. We can make two important observations here. The two overlapping labels can be located only in the intermediate position of~$l_i$, as otherwise we would have overlapping labels at the start or end position of~$l_i$. Additionally, the start and end positions of labels in~$L'$ cannot overlap with $p_i$, as these labels would then overlap with both the start and end positions of $l_i$ and in this way cause an overlap.
	
	Finally, we combine the previous observations geometrically. The start and end positions of~$l_i$ overlap with at most seven labels each, while the intermediate position overlaps with at most four labels. The labels overlapping the intermediate position, and extending into start or end position, should necessarily be one of the three labels intersecting the start or end position of~$l_i$, as otherwise the labelings before or after the transition are not overlap-free. Furthermore, since the movement of $l_i$ is diagonal, first downwards and then rightwards, the top-right corners of the start and end positions are located close to each other. As a result, the labels whose intermediate positions intersect these corners can overlap (see Figures~\ref{fig:naive-overlaps-split-2}a and b). The start positions of the label moving from bottom-right to top-left at the start position, and the label moving from top-left to bottom-right at the end position overlap, and only one of the two labels can exist simultaneously in an overlap-free labeling. Similarly, the end positions of the label moving from top-left to bottom-right at the start position, and the label moving from bottom-right to top-left at the end position overlap, and again only one of them can exist at the same time. Thus we have twelve labels in total for the start and end position of~$l_i$, and two labels intersecting only the intermediate position, leading to at most fourteen labels overlapping~$l_i$.
	
	To conclude the proof, observe that a diagonal movement of~$l_i$ leads to more overlaps than a movement to an adjacent position. When the start and end positions are adjacent, each position can overlap only with two end and start positions respectively. However, since the corners of the start and end positions are no longer as close as before, each corner of, for example, the start position, not adjacent to the end position, can now have both labels whose intermediate position overlaps that corner --- the geometric argument from the diagonal case no longer applies. This leads to six overlaps per start and end position. As there is no intermediate position, the total number of overlaps for~$l_i$ is at most twelve in this case. 
\end{proof}

\begin{figure}
	\centering
	\includegraphics[page=4]{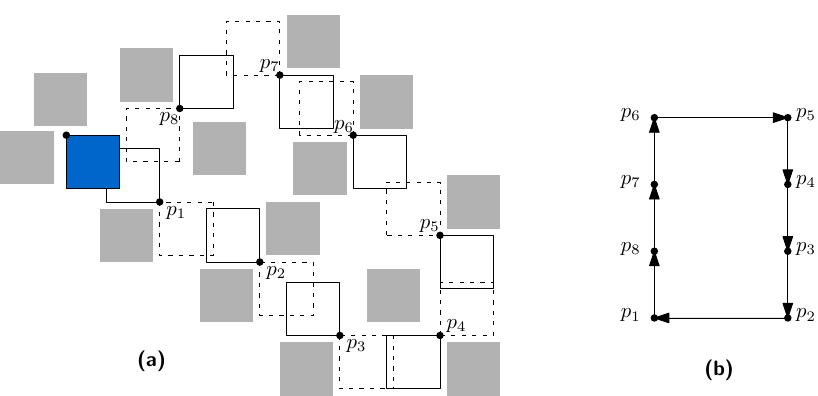}
	\caption{Construction showing that the bound of $n + m$ overlaps (here: $n=8$ and $m=1$) from \cref{thm:dag} is tight. \textbf{\textsf{(a)}} The blue label is added in this transition and forces $n + m$ inevitable overlaps during movement. Gray labels are stationary. \textbf{\textsf{(b)}} The corresponding movement graph for the instance in \textbf{\textsf{(a)}}.}
	\label{fig:dag_n_m_overlaps}
\end{figure}

\subsection{DAG-based Transitions}
\label{sec:dag}
To refine the naive approach, we model dependencies between movements in a \emph{movement graph}, and use it to order movements and avoid certain overlaps.
\begin{definition}[Movement graph]
	\label{def:movement-graph}
	Let $\mathcal{M} = \{\TransitionStyle{M_1}, \dots, \TransitionStyle{M_n}\}$ be a set of movements. 
	Create for each movement $\TransitionStyle{M_i} \in \mathcal{M}$ a vertex $v_i$, and create a directed edge from $v_i$ to $v_j$, $\DirectedEdge{v_i}{v_j}$, if some intermediate or end position of \TransitionStyle{M_j} overlaps with the start position of \TransitionStyle{M_i}, or the end position of \TransitionStyle{M_j} overlaps with some intermediate position of \TransitionStyle{M_i}: In both cases movement~$M_i$ should take place before movement~$M_j$.
	If intermediate positions of \TransitionStyle{M_i} and \TransitionStyle{M_j} overlap, create the edge $\DirectedEdge{v_i}{v_j}$, $i < j$.
	This results in the movement graph $G_{\mathcal{M}}$ (see \cref{fig:dag_n_m_overlaps}b).
\end{definition}
A \emph{feedback arc set} in a movement graph is a subset of edges that, when removed, breaks all cycles, resulting in a directed acyclic graph (DAG). We order movements using this DAG.
\begin{theorem}
	\label{thm:dag}
	Movements in \Transition{R M_1 \dots M_n A} can be rearranged such that at most $n + m$ overlaps occur, if $G_{\mathcal{M}}$, with $\mathcal{M} = \{\TransitionStyle{M_1}, \dots, \TransitionStyle{M_n}\}$, has a feedback arc set of size~$m$.
\end{theorem}
\begin{proof}
	By \cref{lem:one-point}, we know that at most one overlap occurs when moving a single label to a free end position.
	This leads to at most $n$ overlaps for $n$ consecutively moving labels, if no label moves to (or through) a position occupied by a label, which starts moving later.
	
	Let $G_{\mathcal{M}}$ be a movement graph with $\mathcal{M} = \{M_1, \dots, M_n\}$. There are two cases: 
	\begin{description}
		\item[Case (1)] If $G_{\mathcal{M}}$ is acyclic, then handling all movements according to any topological ordering of the vertices of $G_{\mathcal{M}}$ produces no additional overlaps. 
		\item[Case (2)] If $G_{\mathcal{M}}$ contains cycles, then overlaps may be inevitable because each label in such a cycle wants to move to or through a position that is occupied by another moving label.
		Moreover, as the movements happen consecutively, one label in this cycle must move first and therefore may cause an overlap. 
		Let $m$ be the smallest number of edges that must be removed to break each cycle in $G_{\mathcal{M}}$, i.e., the size of a minimum feedback arc set~$S$.
		As $G_{\mathcal{M}}$ is cycle-free after removing~$S$, case (1) applies and $m$ additional overlaps suffice. \qedhere
	\end{description}
\end{proof}

We can see with the example instance from \cref{fig:dag_n_m_overlaps}a that this bound is tight.
Furthermore, it is not always necessary to perform all movements consecutively.
We can observe that movements which are unrelated in $G_{\mathcal{M}}$ can be performed simultaneously: when no overlap is possible, there is no edge in $G_{\mathcal{M}}$.

\section{Simultaneous Transitions}
\label{sec:rma}

\begin{figure}
	\centering
	\includegraphics{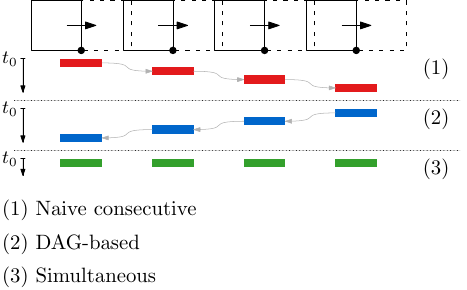}
	\caption{Comparison of possible movement orderings with respect to \GoalOverlapas\space and \GoalDuration.}
	\label{fig:naive_dag_rma}
\end{figure}

\cref{fig:naive_dag_rma} shows three timelines of different transition styles, (1) a naive consecutive transition, (2) a DAG-based transition, and (3) a simultaneous transition. 
All transition styles start at some time $t_0$ and the order of the movements of the labels for (1) and (2) is indicated with gray arrows.
While (1) produces four overlaps and takes four units of time, (2) and (3) produce no overlaps, and (3) only takes a single unit of time. This shows that it is sometimes unnecessary to perform the movements consecutively to minimize overlaps.
In this section, we investigate both how simultaneous movements influence the number of overlaps, and the complexity of minimizing overlaps.

\subsection{An Upper Bound on the Number of Overlaps}
We start our investigation by analyzing the number of overlaps a simultaneous transition can have in the worst case.

\begin{theorem}
	\label{thm:rma}
	In \Transition{R M A} at most $6n$ overlaps can occur, where $n$ is the number of labels that must be moved, and all movements are performed at unit speed.
\end{theorem}
\begin{proof}
	Let $\sigma=1$ denote the side length of a label.
	To show that the total number of overlaps is at most $6n$, we model the overlaps in a graph and consider the neighborhood of individual vertices.
	Let $G$ be a conflict-graph where each vertex $v_i$ corresponds to a label $l_i$.
	If two labels $l_i$ and $l_j$ overlap during the transition, we create an edge ($v_i$, $v_j$), i.e., each edge corresponds to an overlap.
	Observe that each edge is adjacent to at least one moving label since two stationary labels cannot overlap.
	We proceed by evaluating the maximum possible degree of a moving label $l$ and restrict ourselves to a $\sigma$-wide border around the bounding box of the movement area of $l$, that represents the area other labels must touch (before the transition) to overlap with $l$.
	We call this area the \emph{overlapping region} of $l$ and it is illustrated in \cref{fig:overlapping-region}.
	Labels not intersecting the overlapping region of $l$ by construction cannot overlap with $l$.
	We proceed by considering the two possible types of movements for $l$.
	
	\begin{figure}
		\centering
		\includegraphics{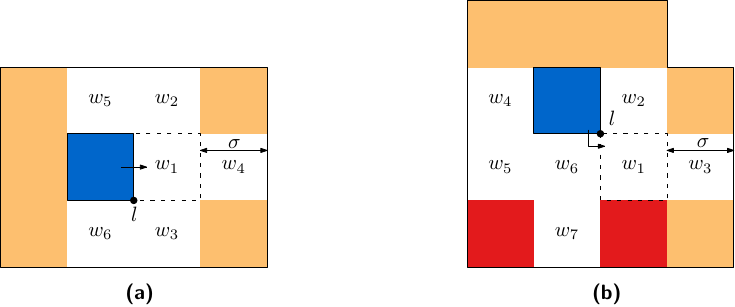}
		\caption{Overlapping regions of the blue label~$l$ depending on its movement type. \textbf{\textsf{(a)}} Non-diagonal movement. \textbf{\textsf{(b)}} Diagonal movement.}
		\label{fig:overlapping-region}
	\end{figure}
	
	\proofsubparagraph{Non-Diagonal Movement of $\boldsymbol{l}$.}
	For a label~$l$ that performs a non-diagonal movement, the overlapping region is illustrated in \cref{fig:overlapping-region}a. The light-orange area in the overlapping region indicates that the start position of a label overlapping $l$ cannot lie solely in this area. If a labels starts in the area behind $l$, then such a label would never overlap with $l$, since labels move simultaneously. The start position of labels overlapping $l$ can neither overlap only the $\sigma \times \sigma$ tiles diagonally adjacent to the end position of $l$, as the end position of those labels would overlap the end position of $l$, for any movement that allows the labels to overlap~$l$.
	
	\begin{figure}[b]
		\centering
		\includegraphics{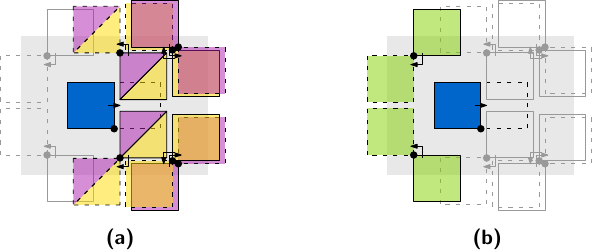}
		\caption{Labels overlapping with the blue label~$l$ that are located on the white tiles. The overlapping region of~$l$ is indicated in gray. \textbf{\textsf{(a)}} White tiles $w_1$ -- $w_4$. \textbf{\textsf{(b)}} White tiles $w_5$ -- $w_6$.}
		\label{fig:rma-non-diagonal-overlaps-right-left}
	\end{figure}
	
	Next consider the remaining (white) area in the overlapping region, and see \cref{fig:rma-non-diagonal-overlaps-right-left} for our construction. Consider first the end position~$w_1$ of $l$ and the three $\sigma \times \sigma$ tiles $w_2,w_3$ and $w_4$ adjacent to it. The total height of $w_1,w_2$ and $w_3$ combined is $3\sigma$, and hence the start positions of at most four labels can be stacked vertically to overlap this area (see the labels with the color \ProofLabelColor{labelcolorpurple} in \cref{fig:rma-non-diagonal-overlaps-right-left}a). Similarly, $w_1$ and $w_4$ have a combined width of $2\sigma$ and height $\sigma$. Since the end position of $l$ is adjacent to the start position of $l$ we can put at most two labels horizontally next to each other in this area, while keeping \Labeling{1} overlap free. However, as the height is $\sigma$ we can stack at most two layers of such labels vertically (the \ProofLabelColor{labelcoloryellow} labels in \cref{fig:rma-non-diagonal-overlaps-right-left}a). As a result, $w_1,w_2,w_3$ and $w_4$ can together overlap with at most six start positions of other labels. Each label results in at most one overlap, and there is a movement direction for each label that achieves such an overlap, as shown in \cref{fig:rma-non-diagonal-overlaps-right-left}a.
	
	Now consider the $\sigma \times \sigma$ tiles~$w_5$ and~$w_6$ above and below the start position of~$l$, respectively. We can place two labels, the ones colored \ProofLabelColor{labelcolorgreen} in \cref{fig:rma-non-diagonal-overlaps-right-left}b, such that their start positions overlap either of~$w_5$ and~$w_6$. For example, for $w_5$ such labels can move diagonally down-left, to overlap $l$. In this case, it is impossible for a label overlapping $w_6$ to both overlap $l$ and have an overlap-free end position. Conversely, we can place one label on~$w_5$ and~$w_6$ each and allow them both to move towards $l$, while ensuring overlap-free end positions (see \cref{fig:rma-non-diagonal-overlaps-right-left}b). Observe that it is impossible to place two labels on both~$w_5$ and~$w_6$ in the latter case, as the vertical positioning that ensures overlap-free end positions of the labels, requires the labels to start farther from $l$. Those labels have to move a vertical distance of at least $\sigma/2$ to reach $l$, and hence also require a horizontal overlap of at least $\sigma/2$ with the start position of $l$, as $l$ will have moved $\sigma/2$ rightwards before the other labels reach the start position of $l$.
	Thus, at most eight labels can overlap with $l$, and consequently the degree of the corresponding vertex is bounded by eight.
	See \cref{fig:rma-overlaps}a for the complete situation.

	\proofsubparagraph{Diagonal Movement of $\boldsymbol{l}$.}
	For diagonal movements, we consider w.l.o.g. the case where $l$ performs a diagonal movement from top-left to bottom-right through the bottom-left corner. The overlapping region enlarges, as shown in \cref{fig:overlapping-region}b.
	We can again eliminate the light-orange areas, as they mark areas that the start position of other labels cannot overlap exclusively, if they should overlap with $l$. As before, these areas are located behind the start position of $l$, and diagonally adjacent to the end position of $l$. We now repeat the process of filling the remaining (white) $\sigma \times \sigma$ tiles, $w_1$ to $w_7$, with start positions for labels that can overlap with $l$ during movement.
	See \cref{fig:rma-diagonal-overlaps-right-left,fig:rma-diagonal-overlaps-middle} for our construction.
	
	We can stack at most three start positions of labels vertically while overlapping~$w_1$ and~$w_2$, as their total height is $2\sigma$. The bottom two start positions overlap with both~$w_1$ and~$w_3$. The total width of these is also $2\sigma$ and hence we can place more start positions for labels overlapping~$w_3$, one in each row. To overlap with $l$, all these labels, which are colored \ProofLabelColor{labelcolorgreen} in \cref{fig:rma-diagonal-overlaps-right-left}a, should start their movement leftwards.
	
	\begin{figure}
		\centering
		\includegraphics{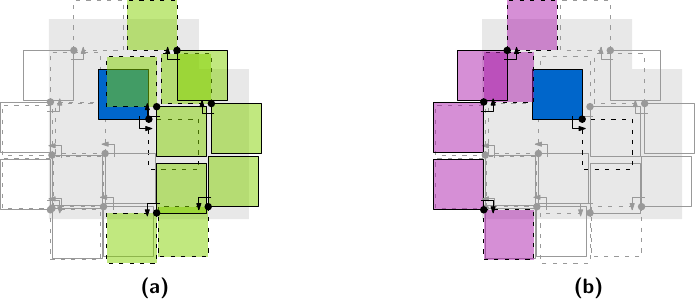}
		\caption{Labels overlapping with the blue label~$l$ that are located on the white tiles. The overlapping region of~$l$ is indicated in gray. \textbf{\textsf{(a)}} White tiles $w_1$ -- $w_3$. \textbf{\textsf{(b)}} White tiles $w_4$ -- $w_5$.}
		\label{fig:rma-diagonal-overlaps-right-left}
	\end{figure}
	
	Next, we can vertically stack at most three start positions to overlap~$w_4$ and~$w_5$ (see the \ProofLabelColor{labelcolorpurple} colored labels in \cref{fig:rma-diagonal-overlaps-right-left}b). If the labels starting in these positions move rightwards first, they also overlap with $l$. This leaves us with~$w_6$ and~$w_7$ below the start position of $l$. Since the start position of $l$ is located directly above these tiles, we can vertically stack at most two start positions of labels overlapping these tiles.
	
	\begin{figure}
		\centering
		\includegraphics{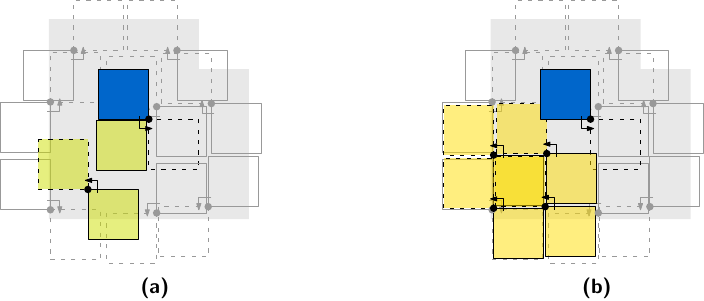}
		\caption{Labels overlapping with the blue label~$l$ that are located on the white tiles $w_6$ -- $w_7$ using different types of labels. The overlapping region of~$l$ is indicated in gray. \textbf{\textsf{(a)}} Using a stationary label at $w_6$. \textbf{\textsf{(b)}} Using a non-stationary labels.}
		\label{fig:rma-diagonal-overlaps-middle}
	\end{figure}
	
	We have two options for~$w_6$ and~$w_7$: either we place a stationary label overlapping~$w_6$, and one moving label overlapping~$w_7$ (see the \ProofLabelColor{labelcolorlime} labels in \cref{fig:rma-diagonal-overlaps-middle}a), side by side, or we try to fit two rows of two start positions for moving labels (see the \ProofLabelColor{labelcoloryellow} labels in \cref{fig:rma-diagonal-overlaps-middle}b). In the former case, we could try to place two moving labels, which first move upwards to overlap $l$. After moving upwards, they must move left, as moving rightwards would result in their end positions overlapping with the end position of $l$. Even if the two labels move to the left, the rightmost of the two labels overlapping~$w_7$ will have its end position overlap the stationary label, thus only a single label can overlap with~$w_7$ in this case, as illustrated in \cref{fig:rma-diagonal-overlaps-middle}a. This case results in a total of at most nine overlapping moving labels, and one stationary label overlapping $l$. The stationary label can be beneficial, as it increases the number of overlaps, without increasing the number~$n$ of moving labels.
	
	In the latter case, where we place two rows of two labels overlapping~$w_6$ and~$w_7$, the labels in the bottom row have to move upwards to overlap $l$. As can be seen in \cref{fig:rma-diagonal-overlaps-middle}b, there is a movement direction for all labels in this case, that results in overlap-free end positions. We now have a total of twelve moving labels that can overlap $l$.
	
	What remains is to consider the red $\sigma \times \sigma$ tiles (see \cref{fig:rma-diagonal-overlaps-middle}b). A label cannot start in the right tile, as it would need to move up- and leftwards to overlap $l$, but this results in overlapping the end position of $l$. Placing a label in the left red tile requires it to first move right and then up, in order to overlap $l$, as the other way around will never result in an overlap. However, the start or end position of this label will overlap with the previously considered labels in~$w_6$ and~$w_7$, and those labels result in multiple overlaps. Thus start positions of moving labels should instead overlap the white area to maximize overlaps.
	
	\proofsubparagraph{Deriving an Upper Bound.} To find an upper bound on the number of overlaps, consider the subgraph $G[V_M]$ induced by the set~$V_M$ of vertices that represent moving labels.
	The degree of each vertex in $G[V_M]$ is bounded by nine, in case the stationary label is present below the start position of $l$, or by twelve, otherwise. Hence we have a degree sum of at most $9n$ or $12n$, respectively, since $\vert V_M \vert = n$.
	By the handshaking lemma, we derive that we have at most $\lceil 4\frac{1}{2}n\rceil$ or $6n$ edges in $G[V_M]$, respectively.
	If we consider the original graph $G$, we can observe that it differs from $G[V_M]$ in terms of edges only by the edges that are incident to a moving label and a stationary label.
	However, as we have seen in one case of the above proof, and due to \cref{lem:one-point}, a moving label can overlap with at most one stationary label.
	Since each of the edges in $E(G) \setminus E(G[V_M])$ is incident to exactly one vertex that represents a moving label, and we have $n$ of such labels, $\vert E(G) \setminus E(G[V_M]) \vert$ is bounded by $n$ and consequently, we have at most $\lceil 5\frac{1}{2}n\rceil$ overlaps with the stationary label present. However, in case there is no stationary label, we arrive at a higher bound of at most $6n$ overlaps.
	The upper bound is tight, as can be seen in \cref{fig:rma-overlaps}b.
\end{proof}

\begin{figure}
	\centering
	\includegraphics{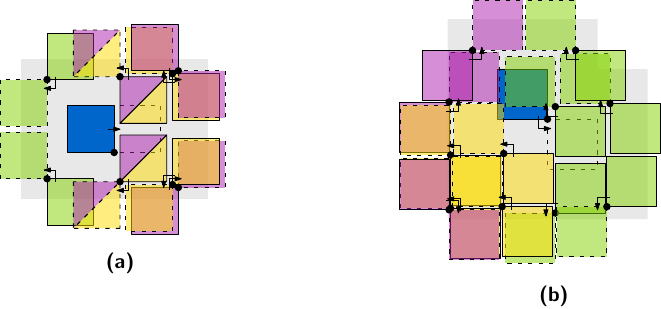}
	\caption{Labels overlapping the blue label~$l$. The overlapping region of~$l$ is indicated in gray. \textbf{\textsf{(a)}} If~$l$ performs a non-diagonal movement. \textbf{\textsf{(b)}} If~$l$ performs a diagonal movement.}
	\label{fig:rma-overlaps}
\end{figure}

\subsection{Complexity of Computing Simultaneous Transitions}
In this section, we show that it is \NPComplete to minimize the number of overlaps in a \emph{weighted} \Transition{RMA}-transition by choosing the direction of diagonal movements.

\begin{definition}[Weighted Transition]
	Let \Transition{\Sigma} be a transition, where \TransitionStyle{\Sigma} denotes an arbitrary transition style of additions, movements, and removals, and let $w$ be a weight function that assigns to each label $l\in L$ a non-negative weight $w(l) \in \mathbb{R}_0^+$.
	A weighted transition \WeightedTransition{\Sigma}{w} performs \Transition{\Sigma}, but when two labels $l_i$ and $l_j$ overlap, a penalty of weight $w(l_i)\cdot w(l_j)$ is introduced. The \emph{total penalty}~$W$ is equal to the sum of penalty weights. 
\end{definition}

\begin{problem}
	\label{prob:rma-minimize-overlaps}
	Given a weighted transition \WeightedTransition{R M A}{w} and $k \in \mathbb R_0^+$, can we assign a movement direction to each diagonal movement such that the total penalty~$W$ is at most $k$?
\end{problem}

\begin{theorem}
	\label{thm:minimizing-rma-np-hard}
	It is \NPComplete to decide whether $W$ is at most $k$ for \WeightedTransition{R M A}{w}.
\end{theorem}
\begin{proof}
	We first argue containment in \NP. Given a movement direction for each label, we can check each pair of labels and compute the penalty of the overlaps. If the total penalty~$W$ is at most $k$ then we have a solution to \cref{prob:rma-minimize-overlaps}, otherwise the given directions do not lead to a valid solution. Hence \cref{prob:rma-minimize-overlaps} is contained in \NP.
	
	To prove \NP-hardness, we build a reduction from \PlanarMonotoneMaxTwoSAT onto \cref{prob:rma-minimize-overlaps}.
	Arguments that \PlanarMonotoneMaxTwoSAT is \NPHard can be found in Buchin, Polishchuk, Sedov, and Voronov~\cite[Theorem 7]{Buchin.2019}.
	The hardness reduction for \PlanarMonotoneMaxTwoSAT uses a reduction from \PlanarMaxTwoSAT, which was shown to be \NPHard in Guibas et al.~\cite{Guibas.1993} by using a planarity-preserving reduction from \PlanarThreeSAT. The incidence graph of \PlanarThreeSAT instances used in the reduction, can be embedded in rectangular grids of polynomial-size~\cite[Section 3]{Cabello.2003}. We therefore also consider only an orthogonal embedding for the incidence graph in our reduction.\footnote{An illustration of such a rectilinear layout and its use in a hardness-proof for a map labeling-related problem can be found in Knuth and Raghunathan~\cite{Knuth.1992}.}
	Furthermore, the incidence graph of a \PlanarMonotoneMaxTwoSAT instance can be embedded in the plane in a way such that all variables are placed on a straight line, all occurrences of the positive variables are on one side and all negative occurrences on the other side of that line~\cite{Berg.2010}.
	Our reduction is inspired by previous proofs on the $\mathsf{NP}$-completeness of point labeling problems in the four-position model from Formann and Wagner~\cite{Formann.1991} and Marks and Shieber~\cite{Marks.1991}.
	
	Given a monotone planar 2-SAT formula $F$ in CNF with $n$ clauses, $C_1$ to $C_n$.\footnote{
		In the literature $n$ usually denotes the number of variables in a \textsc{SAT}-instance; not clauses.
		However, asymptotically this is irrelevant for our reduction:
		A \PlanarMonotoneMaxTwoSAT-instance with $n$ clauses has up to $2n$ distinct variables, thus our reduction is still polynomial in time and size.
	}
	We construct an instance of \WeightedTransition{R M A}{w}, consisting of only labels that require diagonal movement. Each diagonal moving label must choose a movement direction such that there the total penalty $W$ is at most $k$, if and only if we can satisfy at least $n - k$ clauses of $F$.
	
	For each clause, we construct a \emph{clause-gadget} displayed in \cref{fig:gadgets}a.
	The clause-gadget consists of two labels with a weight of one performing a diagonal movement.
	Each label represents one of the variables of the clause, and assigning \emph{``true''} to that variable means that the corresponding label takes the outer path while assigning \emph{``false''} corresponds to taking the respective inwards movement.
	Satisfying the clause can be achieved by assigning ``true'' to at least one of the variables, with the consequence that the potential overlap inside the gadget can be avoided.
	Thus, when the clause cannot be satisfied, both variables are ``false'', and an overlap with penalty one occurs in the corresponding clause-gadget.
	What remains is to ensure that the variable assignment is consistent among all clauses. 
	We use the structure of the planar incidence graph of $F$ to achieve this.
	
	\begin{figure}
		\centering
		\includegraphics{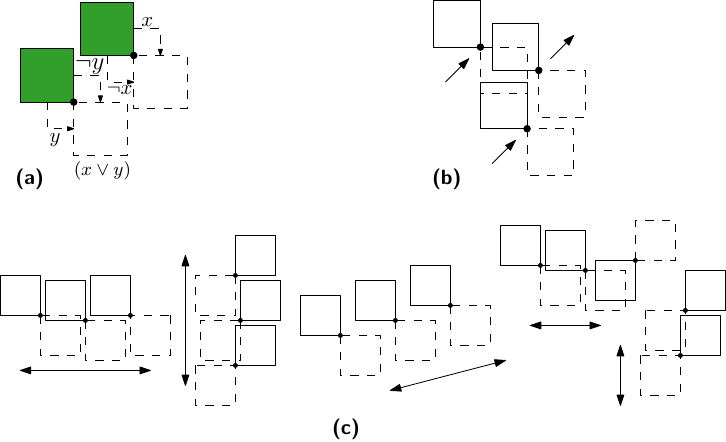}
		\caption{Gadgets in our reduction. White and green labels have weight $n + 1$ and one, respectively. \textbf{\textsf{(a)}} A clause-gadget for the clause $(x \lor y)$. \textbf{\textsf{(b)}} A merge-gadget. \textbf{\textsf{(c)}} Various transportation-gadgets for different use cases.}
		\label{fig:gadgets}
	\end{figure}
	
	First, we introduce \emph{transportation-gadgets}, displayed in \cref{fig:gadgets}c.
	Using these gadgets, we can transport the selection of a literal in a clause-gadget over the plane, i.e., along the edges of our incidence graph.
	By construction, a transportation-gadget ensures that, once one label at the end of the gadget chooses the movement inwards (towards the gadget), the gadget becomes \emph{active} and passes this selection on to the other end.
	This means that, for example, if in the leftmost transportation-gadget in \cref{fig:gadgets}c, the leftmost label makes its diagonal movement clockwise, all labels will make their movement clockwise.
	Similarly, all labels move counterclockwise, if the rightmost label in this gadget performs its movement counterclockwise.
	As we can further see in \cref{fig:gadgets}c, the transportation-gadgets are also able to perform ninety-degree bends while preserving the movement direction.
	This transportation-gadget will be used along the edges of the incidence graph, starting at the vertices representing the clauses, i.e., the clause-gadgets, and leading towards the vertices representing the variables, i.e., the \emph{variable-gadgets}, which we will introduce soon. The number of labels used in a transportation-gadget is proportional to the length of the edge it represents.
	Furthermore, we can slightly adapt the direction of a transportation-gadget by changing the relative offset of the moving labels, which will be useful inside variable-gadgets.
	Finally, each label $l$ of a transportation-gadget has a weight of $w(l) = n + 1$. Thus, any overlap inside a transportation-gadget will result in a penalty larger than $n$ and instantly invalidate the solution, as a valid solution has a total penalty~$W$ of at most $k\leq n$.
	
	To introduce the variable-gadget, we first need a \emph{merge-gadget}, which can merge two independent transportation-gadgets into one transportation-gadget.
	\cref{fig:gadgets}b shows the building blocks of the gadget: it has two incoming ends and one outgoing end, to which transportation-gadgets can attach.
	Once at least one of the transportation-gadgets attached to the incoming ends becomes active, then the transportation gadget at the outgoing end is forced to be active as well, to avoid overlaps inside the gadget.
	Since the labels have a weight of $n + 1$, an overlap inside this gadget would result in a total penalty of at least $(n + 1)^2$, and hence the outgoing transportation-gadget it forced to be active, if one of the incoming transportation-gadgets is active. 
	
	To ensure that we do not encounter invalid variable assignments, we use variable-gadgets.
	A variable-gadget is a combination of merge-gadgets and a variation of the clause-gadget, and will be placed at each vertex of the incidence graph representing a variable in $F$.
	
	A variable-gadget for variable $x_i$ has many incoming transportation-gadgets for occurrences of $x_i$ and $\lnot x_i$ in clauses. These transportation gadgets are pairwise merged until we have only a single transportation-gadget for $x_i$ and one for $\lnot x_i$. Since all clauses with only positive variables are placed on one side of the variable-gadget and the clauses with only negative variables are on the other side, it is guaranteed that we can merge only transportation-gadgets that transport the same information, i.e., we never merge transportation-gadgets for $x_i$ and $\lnot x_i$.
	The two remaining transportation-gadgets, one at each side, are connected using a clause-gadget with adapted weights: the involved labels all have a weight of $n + 1$.
	This setup ensures that if at least one transportation-gadget per side was active, the two transportation-gadgets that arise from the merging process are active as well, and in the middle of the variable-gadget we get an overlap with a penalty of $(n + 1)^2$. 
	
	The number of labels used in the variable-gadget for~$x_i$ is proportional to the number of occurrences of $x_i$ and $\lnot x_i$.
	Observe that merge-gadgets form two rooted binary trees connected at their roots by the clause-gadget. The occurrences of $x_i$ and $\lnot x_i$ form the leaves of the two trees (see \cref{fig:completereduction-example}).
	Hence, if there are $o$ and $o'$ occurrences of $x_i$ and $\lnot x_i$, respectively, then the variable-gadget for~$x_i$ consists of $o + o' - 2$ merge-gadgets, one clause gadget, and few labels to connect these gadgets.
	
	The total number of labels in the reduced instance is proportional to the size of the embedded incidence graph. We require two labels per clause, transportation-gadgets use a number of labels proportional to the length of the edge they represent, and each variable-gadgets uses a number of labels proportional to the number of (positive and negative) occurrences of that variable. The total number of labels is hence polynomial.
	
	We will now prove the correctness of the reduction. First assume that $n-k\geq 0$ clauses of $F$ can be satisfied with a consistent variable assignment.
	We should encounter a total penalty~$W$ of at most $k$ in the reduced instance, which can originate only from overlaps of weight one in $k$ clause-gadgets.
	Consider the clause-gadgets, and assign movement directions corresponding to the variable assignment in $F$: outwards if the occurrence of a variable is ``true'' and inwards if the occurrence is ``false''. Only in the $k$ clause gadgets that are not satisfied, do we get an overlap with penalty one, since all clause-gadget labels have weight one. The total penalty so far is $k$. The transportation for the outwards moving labels in clause-gadgets will be active, and will be transported, through merge-gadgets to the middle of each variable-gadget. Since the variable assignment is consistent, each variable~$x_i$ will be set to ``true'' for only its positive or its negative occurrences. Thus the movement directions chosen in the clause-gadgets will result in no additional overlaps in the variable gadgets. The total penalty~$W$ is hence $k\leq n$.
	
	Finally, assume that in the reduced instance there is a movement direction for each label, such that $W=k\leq n$.
	This means that there can be overlaps only in the clause gadgets, as overlaps outside of the clause-gadgets would result in a total penalty~$W$ larger than $n$. We show that there is an assignment of the variables in $F$ such that exactly those clauses are satisfied, which correspond to clause-gadgets without overlaps.
	
	The variation of the clause-gadget inside each variable-gadget can now be in one of two states:
	either one of two labels, corresponding to the positive and negative occurrences of the variable, chose the outwards movement direction, or both labels chose the outwards direction. In the former case, if the side of the positive occurrences chose an inwards movement, we set the corresponding variable to ``true'', otherwise we set the variable to ``false''. The clause-gadgets connected to the side that chose the inwards movement will be able to chose an outwards movement for the corresponding variable in the clause-gadget, and this will not result in any overlaps in transportation- or merge-gadgets. Thus either the clauses with all positive or all negative occurrences of such a variable will be satisfied using this variable assignment, mimicking the clauses without overlaps. 
	
	We can now observe that all $k$ clause-gadgets without overlaps are already accounted for, and hence $n-k$ clauses in $F$ are already satisfied: in the latter case for the variable-gadget, both positive and negative occurrences have an outward movement direction in the variable-gadget, and hence an inwards movement direction in the corresponding clause-gadgets. This means that the remaining variable-gadgets are not connected to clause-gadgets in which they provide the outgoing movement direction that prevents an overlap. These variable-gadgets are in a ``don't care''-state, and we can choose either assignment of the variable without changing the already satisfied clauses. Thus only the $k$ clauses corresponding to clause-gadgets with overlaps are not satisfied, and hence $n-k\geq 0$ clauses in $F$ are satisfied.
	
	Combining the above arguments, we can conclude that we have reduced the decision version of \PlanarMonotoneMaxTwoSAT onto \cref{prob:rma-minimize-overlaps}, and therefore, it is \NPComplete to decide whether $W$ is at most $k$ for \WeightedTransition{R M A}{w}.
\end{proof}

\begin{figure}
	\centering
	\includegraphics{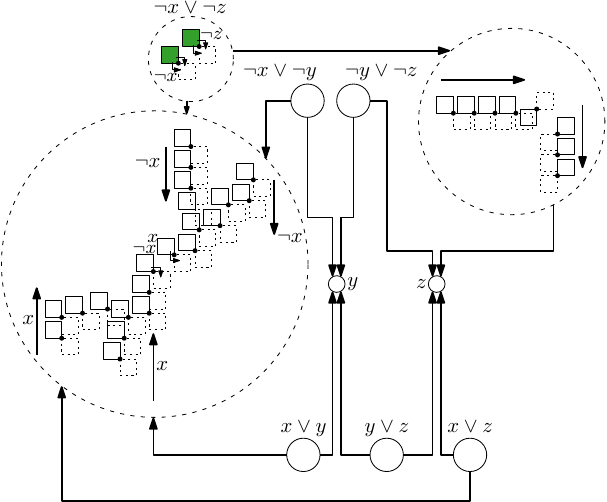}
	\caption{Reduced instance for the formula $F = (\lnot x \lor \lnot z) \land (\lnot x \lor \lnot y) \land (\lnot y \lor \lnot z) \land (x \lor y) \land (y \lor z) \land (x \lor z)$. 
		The weight of white and green labels is $n + 1$ and one, respectively.}
	\label{fig:completereduction-example}
\end{figure}

\cref{fig:completereduction-example} shows an example of the complete setup.
There we reduce the formula $F = (\lnot x \lor \lnot z) \land (\lnot x \lor \lnot y) \land (\lnot y \lor \lnot z) \land (x \lor y) \land (y \lor z) \land (x \lor z)$ onto our problem.
Circles with solid borders indicate the individual parts of the formula, while the dashed circle shows part of a transportation gadget. The big empty circles represent the clauses and small empty circles represent the variables.
The arrows outside circles represent the transportation-gadgets that connect the individual parts according to the direction of the arrows. 
As there exists no variable assignment which satisfies $F$, we cannot achieve $W = k = 0$ but must encounter at least one overlap (and hence $W \geq 1$).
This overlap occurs, for example, in the clause gadget for $(\lnot x \lor \lnot z)$ to achieve $W=1$.
Note that if we would try to resolve this overlap by, for instance, setting $x$ to true and false at the same time, an overlap with penalty $(n+1)^2=49$ would occur, for example, in the variable-gadget for the variable $x$.

\section{Case Study: Transition Quality in Real-World Scenarios}
\label{sec:case-study}
To complement our theoretical investigations in the previous sections, we implemented a prototype to analyze how the transition styles perform in a practical setting.
Our prototype labels spatiotemporal data (\cref{sec:labeling-model}) on an interactive map (\cref{subsec:case-study-implementation}) at a given time of interest.
Interacting with the prototype by means of panning or zooming the map or altering the time of interest, the shown data points change, which we visualize by applying our transition styles on a set of axis-aligned labels.
In our case study, we consider different (user-generated) real-world data (\cref{subsec:case-study-datasets}) and consider different interaction schemes with the map.
For each of them, we measure the computation time of the labeling (\cref{subsec:case-study-results-backend}) as well as the number of overlaps and the time to compute and perform a transition (\cref{subsec:case-study-results-frontend}).
The source code and videos of the different transition styles are online publicly available~\cite{OSF}.

\subsection{Dynamic Labeling Model}
\label{sec:labeling-model}
The data we use have spatiotemporal properties: each point $p$ has a location and a time associated to it.
Starting from this point in time, we consider $p$ (and its associated data) \emph{relevant} for a given, data-dependent, time span.
The relevant points at a particular \emph{time of interest} will form the set $P$ of points that we want to label.
Changes to the time of interest (dynamically) alter the set~$P$ of relevant points through additions and removals.
Furthermore, in our implementation not all points in $P$ will be in view at all times:
For example, when the user zooms in on a particular part of the map, some points will be outside the view port. In such cases, we label only the subset $S \subseteq P$ of points that are inside the view port.

\subsection{Implementation Details}
\label{subsec:case-study-implementation}

\begin{figure}
	\centering
	\includegraphics[width=.95\linewidth]{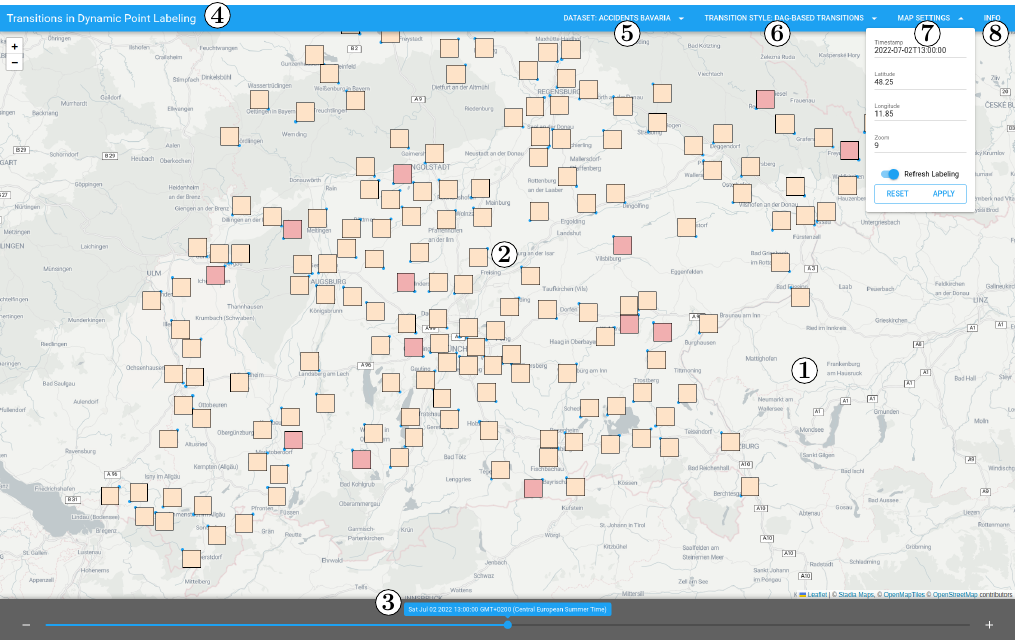}
	\caption{Screenshot of the prototype.}
	\label{fig:prototype-screenshot}
\end{figure}

The prototype computes a labeling in the four-position model of the relevant points~$S$ that are in the view port.
\cref{fig:prototype-screenshot} shows a screenshot of our prototype.
The main view area (\circled{1}), in the center of the screen, shows a map and a labeling overlay.
In the following, we state the location of a point and the coordinates of the map center using their respective \emph{longitude} and \emph{latitude} values assuming \emph{equirectangular projection}.
In general, the experiments could be conducted with any map projection and we suspect that different projections could lead to different but comparable results.
The map contains blue dots (\circled{2}) indicating the locations of the points in $S$ that are labeled.
Below the map, the time slider (\circled{3}) shows the currently selected time of interest, which must be between the earliest and latest timestamp associated to a point in the dataset.
The top bar (\circled{4}) allows the user to change the underlying dataset (\circled{5}), transition style (\circled{6}), alter the state of the map (\circled{7}) using the map settings, or display additional information about the prototype and the used data (\circled{8}).

The user can interact with the prototype by means of \emph{panning} and \emph{zooming} the map, as well as changing the time of interest by using the \emph{time slider}.
Panning is done by dragging the map using the mouse, while zooming is controlled using either the mouse wheel or the zoom indicators in the upper-left corner of the map.
Our implementation supports discrete, step-wise zooming, where each (integer) step is called a \emph{zoom level}.
The time of interest is changed by dragging the indicator in the time slider, or using the $+$ and $-$ buttons on either side of the slider that increase or decrease the time of interest by $1\%$ of the overall possible time span.
Panning and zooming change the subset $S\subseteq P$ of points in view, while changes to the time of interest alter the relevant points~$P$.

\subsubsection{Computing a Labeling}
For a given subset $S\subseteq P$ of points that are both relevant and in view, we compute a labeling as follows.
We create a conflict graph of labels for $S$ and use a simple greedy approximation algorithm for a \textsc{Maximal Independent Set} $I$ in this graph: iteratively add a minimum-degree vertex to $I$, that shares no edge with vertices in $I$.

When the user now interacts with the map, we again perform the same algorithm to find a new labeling, but we use two simple heuristics to improve the stability of the labeling.
The first heuristic is based on desideratum D1 from Been et al.~\cite{Been.2006}.
Among other things, it proposes that the same labels should remain visible when zooming.
To achieve this, we remove all neighbors of the previously shown labels in the conflict graph, to ensure they are picked again.

The second heuristic attempts to prevent unnecessary changes in the labeling:
Let $I_1$ be the subset of labeled points that remained relevant and visible after panning/zooming/time change, and let $I_2$ be the newly computed set of points to be labeled.
If $I_2$ is less than $2\%$ larger than $I_1$, then we simply keep the labeling of $I_1$ instead of swapping to a labeling of $I_2$.

When the subset $S\subseteq P$ of relevant points in view is changed, through panning or zooming, or when $P$ is dynamically altered by changes in the time of interest, our prototype will trigger a transition.

\subsubsection{Transition Styles}
Let \Labeling{1}, \Labeling{2}\ be the computed labelings before and after the change, respectively.
Our prototype supports naive, DAG-based, and simultaneous transitions for \Transition{}.
In each of these transition styles, a movement of a single label has a duration of one second and a removal or addition of a label is animated as a fade-out/fade-in effect with a duration of half a second.
A diagonal movement is split up into two non-diagonal movements with a duration of one second each, starting with the horizontal movement.
Recall that 400 milliseconds turned out to be the critical duration after which moving objects on a map can be recognized~\cite{KFN.Dmp.2016}.
\subparagraph*{Naive Transitions.}
Movements in this transition are performed consecutively in arbitrary order.
Their order is based on the order in which we recognize the need for a movement.

\subparagraph*{DAG-based Transitions.}
The movements in DAG-based transitions are also performed consecutively, though ordered according to a topological ordering of the movement graph~$G_{\mathcal{M}}$.
If $G_{\mathcal{M}}$ contains cycles, we remove the vertex with the lowest in-degree and first move the label of the removed vertex.
Additionally, in order to reduce the duration times of the transition, we perform unrelated movements in $G_{\mathcal{M}}$ simultaneously.

\subparagraph*{Simultaneous Transitions.}
The movements in this transition are all performed simultaneously, immediately after the removals.
The direction of diagonal movement is not optimized for minimum overlaps. Instead, we move horizontally first, to create a more uniform transition.

\subsubsection{Implementation} 
The prototype is a three-tier-architecture, consisting of a graph-\emph{database} (Neo4j) storing the data points together with their potential label candidates and overlaps, an \emph{application tier} (Java Play Framework) computing the (new) labeling, and the \emph{presentation tier} (Vue.js, Leaflet, and GreenSock) with which the user can interact and which visualizes the transitions.

For this case study we used a standard off-the-shelf laptop with an Intel\textregistered Core\texttrademark\space i7-1355U with 32 Gigabyte RAM running Ubuntu 22.04.4 LTS (64 Bit).
The prototype was implemented using Neo4j community edition version 5.21.0, Java openjdk 17.0.11, Java Play Framework version 3.0.4, node version 20.15.0, Vue.js version 3.4.31, Leaflet version 1.9.4, and GreenSock version 3.12.5. 
The presentation tier was accessed using Firefox version 128.0 (64-bit) on an external 24'' monitor with a resolution of 1920×1200px.

\subsection{Datasets and Interaction Scenarios}
\label{subsec:case-study-datasets}
We use three datasets each based on real-world data obtained from accident statistics, Tweets\footnote{At the time of creating the dataset the social-media platform was known as \emph{Twitter} and not \emph{X}. To avoid confusion, we decided to still name the dataset \emph{Twitter}.}, or weather reports.
For each dataset, we select a range of zoom levels on which we display the labels together with corresponding label sizes such that we can display an appropriate amount of information inside the labels and still ensure that a reasonable number of labels can fit on the screen.
Although our labels are thus not necessarily squares of side length one, all labels of the same dataset and for a given zoom level and view port have the same size.
Hence, our theoretical results derived in \cref{sec:consecutive,sec:rma} carry over to this model by scaling horizontally and/or vertically.

In the following, we describe each dataset in more detail and summarize the properties of the datasets in \cref{tab:datasets-properties}.

\begin{table}
	\centering
	\caption{Properties of the three datasets.}
	{
		\begin{tabular}{l c r r r r} 
						\toprule
						& & & & \multicolumn{2}{c}{Zoom level}\\
						\cmidrule{5-6} Name & Labels & \# Points & Relevant for & Min & Max\\ 
						\midrule
						\DatasetAccidents (Bavaria) & Squares & 8252 & 1 week & 8 & 20\\
						\DatasetAccidents (Berlin) & Squares & 1569 & 1 week & 9 & 20\\
						\DatasetTwitter & Rectangles & 25000 & 3 hours & 7 & 16\\
						\DatasetWeather & Squares & 1171 & 2 days / 2 weeks & 11 & 20\\
						\bottomrule
		\end{tabular}
	}
	\label{tab:datasets-properties}
\end{table}

\subparagraph*{Accidents.}
The dataset \DatasetAccidents is based on road accidents reported by the German Office for Statistics for the year 2022~\cite{SFL2024,ITNW2024}.
Each accident corresponds to one data point, and we use the information provided by the Office for Statistics to determine the time and place the accident happened.
We use axis-aligned square labels to show an accident on the map.
The label shows a colored icon indicating the type and severity of the accident.
An accident is relevant for one week, and we created two variants of this dataset.
The first variant contains accidents from Bavaria with at least one major injured or killed person and the second variant contains accidents in Berlin where at least one pedestrian was involved.
Labels are visible between zoom levels eight and twenty, and between nine and twenty for the former and latter variant, respectively.

\subparagraph*{Twitter.}
For the dataset \DatasetTwitter, we queried 100,000 geotagged tweets related to the COVID-19 pandemic during the month of May 2021; see \cref{tab:query-keywords}. After filtering and cleaning this dataset, 99,982 usable tweets remained, out of which we randomly selected 25,000.

\begin{table}
	\centering
	\caption{The \texttt{Keywords} and \texttt{\#Hashtags} we used to query the tweets.}
	\ttfamily
	\small
	{\resizebox{\linewidth}{!}{
		\begin{tabular}{l l l l l l l}
						\toprule
						corona & \#corona & covid & \#covid & covid19 & \#covid19 & covid-19\\
						vaccine & \#vaccine & quarantine & \#quarantine & lockdown & \#lockdown & moderna \\
						outbreak & \#outbreak & immune &
						\#immune & immunity & \#immunity & biontech\\
						who & \#who & desease & \#desease & masks & \#masks & pfizer \\
						pandemic & \#pandemic & mutation & \#mutation & ffp2 & \#ffp2 \\
						\#StaySave & \#StayAtHome & \multicolumn{2}{l}{\#FlattenTheCurve} & astrazeneca & \multicolumn{2}{l}{johnson \& johnson}\\
						\bottomrule
		\end{tabular}
	}}
	\label{tab:query-keywords}
\end{table}

A tweet will be represented as a point~$p\in P$, and its label as an axis-aligned rectangle.
For the spatial and temporal property of the point, we use the metadata attached to the tweet.
If the location assigned to a tweet is not a point but a (rectangular) area on the map, we choose an arbitrary location inside this area to prevent artificial cluster creation.
A tweet is relevant for three hours and its label visible between zoom levels seven and sixteen.

\subparagraph*{Weather.}
The third dataset, \DatasetWeather, is based on observations of local weather and soil conditions reported by users to GeoSphere Austria's online service wettermelden.at~\cite{GeoSphereAustria2024}.
Recall that we consider this service a perfect real world applications that would benefit from implementing the transitions proposed in this paper.

The dataset contains the reports within Vienna metropolitan area from July and August 2023.
Each report corresponds to one point in the dataset and we use  the same labels as wettermelden.at.
However, as they are round, we inscribe them into an axis-aligned square to fit our theoretical model.
We consider a report relevant for two days or weeks, depending whether it reports the condition of the weather or the soil, respectively.
Labels are visible between zoom levels eleven and twenty.

\subsubsection{Interaction Scenarios}
In our case study, we use our prototype to simulate eighteen interaction settings in nine scenarios across the three datasets.
The different scenarios we use in our case study are described in \cref{tab:case-study-states}.
In the first setting, we interact with the prototype by applying the following scenario-dependent sequence of interactions.

\begin{table}
	\centering
	\caption{The different scenario states of the case study.}
	{\resizebox{\linewidth}{!}{
		\begin{tabular}{l l c r r r r} 
						\toprule
						& & & \multicolumn{2}{c}{Map center} & &\\
						\cmidrule{4-5} Dataset & Scenario Name & Interaction & Longitude & Latitude & Zoom level & Time of interest\\ 
						\midrule
						\multirow{3}{*}{\DatasetAccidents} & Bavaria (Munich) & (a) & 11.85 & 48.25 & 9 & 2022-07-02T13:00:00\\
						& Bavaria (Nuremberg) & (b) & 11.05 & 49.53 & 11 & 2022-10-09T12:00:00\\
						& Berlin & (c) & 13.60 & 52.50 & 11 & 2022-03-18T15:00:00\\
						\midrule
						\multirow{4}{*}{\DatasetTwitter} & Italy & (a) & 14.45 & 41.30 & 7 & 2021-05-29T13:20:00\\
						& Leeds & (b) & -1.60 & 53.44 & 7 & 2021-05-29T13:00:00\\
						& New Delhi & (a) & 71.18 & 30.20 & 7 & 2021-05-29T08:30:00\\
						& S\~ao Paulo & (c) & -45.00 & -20.65 & 7 & 2021-05-29T02:30:00\\
						\midrule
						\multirow{2}{*}{\DatasetWeather} & June & (a) & 16.38 & 48.21 & 12 & 2023-06-17T08:00:00\\
						& July & (b) & 16.38 & 48.21 & 13 & 2023-07-11T16:00:00\\
						\bottomrule
		\end{tabular}
	}}
	\label{tab:case-study-states}
\end{table}

\begin{description}
	\item[Interaction (a):]
	(1) Increase the time of interest by thirty minutes for the dataset  \DatasetTwitter~/ three days for the dataset \DatasetAccidents~/~one day for the dataset \DatasetWeather, (2) zoom in by one zoom level with the help of the zooming indicators, (3) increase the latitude of the map's center by $0.28$ for the datasets \DatasetTwitter and \DatasetAccidents~/~$0.05$ for the dataset \DatasetWeather using the map settings, and (4) increase the time of interest with the $+$ button next to the time slider.
	\item[Interaction (b):] Interaction~(a) in reverse order:
	(1) Increase the time of interest with the $+$ button next to the time slider, (2) increase the latitude of the map's center by $0.28$ for the datasets \DatasetTwitter and \DatasetAccidents~/~$0.05$ for the dataset \DatasetWeather using the map settings, (3) zoom in by one zoom level with the help of the zooming indicators, and (4) increase the time of interest by thirty minutes for the dataset \DatasetTwitter~/~three days for the dataset \DatasetAccidents~/~one day for the dataset \DatasetWeather.
	\item[Interaction (c):]
	(1) Zoom in by one zoom level with the help of the zooming indicators, (2) \emph{decrease} the time of interest with the $-$ button next to the time slider, (3) \emph{decrease} the \emph{longitude} of the map's center by $1.7$ for the dataset \DatasetTwitter~/~$0.1$ for the dataset \DatasetAccidents\footnote{There is no scenario where we use the dataset \DatasetWeather and perform the interaction sequence~(c).} using the map settings, and (4) increase the time of interest by \emph{twenty} minutes for the dataset \DatasetTwitter~/~two days for the dataset \DatasetAccidents.
	
\end{description}

In the second setting, we continuously increase the time of interest by repeatedly clicking fifteen times the $+$ button next to the time slider.

\subsection{Measuring Labeling Computation Time}
\label{subsec:case-study-results-backend}
We measured the (wall-clock) time spend to compute a labeling and summarize the results in \cref{tab:case-study-results-backend}.
\cref{tab:case-study-results-backend} furthermore shows the running time of the most important components.
There, we can see that for many scenarios the majority of the time  (between 70\% and 95\%) is spent on querying the database.
This can be seen as the bottleneck of our labeling algorithm since the remaining parts run together in less than 85ms.
Even though efficiently computing a (good) labeling is not part of our investigation, we want to underline with these results that our naive labeling algorithm is capable of providing a decent result in under two seconds.
Thus, we think that in terms of efficiency, it is very well suited to simulate a real-world application.

\begin{sidewaystable}
	\centering
	\caption{Running times in milliseconds to compute a labeling averaged over all interactions with the prototype in a given scenario and setting. QT = Query Time Database, CG CT = Conflict Graph Creation Time, \textsc{MIS} CT = \textsc{Maximal Independent Set} Computation Time.}
	{
			\begin{tabular}{clcrrrrrrrrr}
								\toprule
								& & & \multicolumn{2}{c}{Conflict Graph (CG)} & \multicolumn{3}{c}{Changes in the Labeling} & \multicolumn{4}{c}{Running times}\\
								\cmidrule{4-12}
								Dataset & Scenario & Setting & \#Vertices & \#Edges & \#Added & \#Removed & \#Moved & QT & CG CT & \textsc{MIS} CT & Total \\
								\midrule
								\multirow{6}{*}{\rotatebox{90}{\DatasetAccidents}} & \multirow{2}{*}{Bavaria (Munich)} & 1 & 443.00 & 1495.75 & 35.25 & 56.50 & 24.25 & 168.49 & 1.21 & 6.18 & 451.09 \\
								& & 2 & 596.53 & 2384.53 & 74.40 & 77.80 & 29.80 & 145.76 & 1.68 & 8.45 & 461.21 \\
								& \multirow{2}{*}{Bavaria (Nuremberg)} & 1 & 42.00 & 74.00 & 4.75 & 8.75 & 0.50 & 97.05 & 0.13 & 0.31 & 139.68 \\
								& & 2 & 70.67 & 140.80 & 9.00 & 9.33 & 1.60 & 95.49 & 0.18 & 0.44 & 155.23 \\
								& \multirow{2}{*}{Berlin} & 1 & 97.00 & 177.25 & 4.50 & 5.25 & 1.75 & 55.63 & 0.31 & 0.83 & 133.40 \\
								& & 2 & 100.00 & 318.87 & 13.07 & 13.27 & 4.80 & 65.76 & 0.34 & 0.83 & 175.81 \\
								\midrule
								\multirow{8}{*}{\rotatebox{90}{\DatasetTwitter}} &\multirow{2}{*}{Italy} & 1 & 18.00 & 37.50 & 0.25 & 2.50 & 0.00 & 1144.74 & 0.07 & 0.37 & 1200.63 \\
								& & 2 & 46.40 & 103.00 & 2.33 & 2.33 & 1.20 & 1212.58 & 0.09 & 0.48 & 1271.65 \\
								& \multirow{2}{*}{Leeds} & 1 & 595.00 & 15331.50 & 18.00 & 21.25 & 6.75 & 1000.30 & 8.11 & 11.42 & 1250.74 \\
								& & 2 & 925.07 & 29461.40 & 28.80 & 28.80 & 10.00 & 1204.41 & 14.61 & 19.35 & 1442.73 \\
								& \multirow{2}{*}{New Delhi} & 1 & 188.00 & 6438.00 & 2.00 & 8.75 & 2.00 & 973.71 & 2.27 & 4.73 & 1047.59 \\
								& & 2 & 359.47 & 11081.87 & 10.40 & 10.33 & 4.47 & 1023.28 & 4.50 & 7.41 & 1098.45 \\
								& \multirow{2}{*}{S\~ao Paulo} & 1 & 72.67 & 462.50 & 1.42 & 3.00 & 1.25 & 853.78 & 0.31 & 1.20 & 917.59 \\
								& & 2 & 51.73 & 345.93 & 2.47 & 3.33 & 0.80 & 830.82 & 0.18 & 0.30 & 856.50 \\
								\midrule
								\multirow{4}{*}{\rotatebox{90}{\DatasetWeather}} & \multirow{2}{*}{June} & 1 & 961.00 & 49835.75 & 6.50 & 15.50 & 4.00 & 813.02 & 26.83 & 33.84 & 1184.57 \\
								& & 2 & 982.13 & 48273.47 & 7.73 & 8.40 & 4.53 & 738.35 & 26.35 & 26.13 & 1014.66 \\
								& \multirow{2}{*}{July} & 1 & 603.00 & 33531.00 & 1.75 & 13.75 & 0.75 & 1023.93 & 18.17 & 16.84 & 1441.30 \\
								& & 2 & 979.47 & 80133.00 & 5.27 & 5.80 & 1.53 & 1054.53 & 41.56 & 39.77 & 1410.85 \\
								\bottomrule
			\end{tabular}
	}
	\label{tab:case-study-results-backend}
\end{sidewaystable}

\subsection{Measuring Transition Time and Number of Overlaps}
\label{subsec:case-study-results-frontend}
The main focus of this case study is to evaluate the performance of our transition styles in a practical application.
Therefore, we measured for each setting and each transition style the time required to compute and perform the transition and the number of overlaps that occurred during that transition.

\begin{sidewaystable}
	\centering
	\caption{Maximum ($\max$) and average ($\varnothing$) computation time in milliseconds for each transition style in each setting. Best average computation times per row are highlighted bold. Recall that we perform in DAG-based transitions unrelated movements simultaneously.}
	{\begin{tabular}{clcrrrrrr}
						\toprule
						&& & \multicolumn{2}{c}{Naive transitions} & \multicolumn{2}{c}{DAG-based transitions} & \multicolumn{2}{c}{Simultaneous transitions}\\
						\cmidrule{4-9}
						Dataset & Scenario & Setting & \multicolumn{1}{c}{$\max$} & \multicolumn{1}{c}{$\varnothing$} & \multicolumn{1}{c}{$\max$} & \multicolumn{1}{c}{$\varnothing$} & \multicolumn{1}{c}{$\max$} & \multicolumn{1}{c}{$\varnothing$} \\
						\midrule
						\multirow{6}{*}{\rotatebox{90}{\DatasetAccidents}}&\multirow{2}{*}{Bavaria (Munich)}& 1 & 21 & 12.00 & 10 & 7.00 & 10 & \textbf{6.00} \\
						& & 2 & 12 & 7.13 & 9 & 6.20 & 10 & \textbf{5.67} \\
						&\multirow{2}{*}{Bavaria (Nuremberg)}&1 & 4 & 2.75 & 3 & 2.00 & 3 & \textbf{1.75}\\
						&& 2 & 5 & 2.67 & 8 & 3.07 & 6& \textbf{2.33} \\
						&\multirow{2}{*}{Berlin}&1  & 5 & \textbf{2.50} & 6 & 4.00 & 5 & \textbf{2.50} \\
						& & 2 & 5 & \textbf{2.73} & 16 & 4.07 & 8 & 3.27 \\
						\midrule
						\multirow{8}{*}{\rotatebox{90}{\DatasetTwitter}} &\multirow{2}{*}{Italy}&1 & 2 & \textbf{1.25} & 7 & 2.75 & 5 & 2.25 \\
						&&2 & 4 & \textbf{1.80} & 4 & 2.13 & 9 & 2.00 \\
						&\multirow{2}{*}{Leeds}&1 & 6 & 4.00 & 10 & 5.00 & 5 & \textbf{2.75} \\
						&&2 & 7 & 3.80 & 7 & \textbf{3.53} & 8 & 4.13 \\
						&\multirow{2}{*}{New Delhi}&1 & 9 & 4.00 & 6 & \textbf{3.50} & 12 & 4.75 \\
						&&2 & 11 & 3.40 & 8 & \textbf{3.00} & 9 & 3.13 \\
						&\multirow{2}{*}{S\~ao Paulo}&1 & 5 & 3.00 & 3 & 1.75 & 1 & \textbf{0.50} \\
						&&2 & 8 & 1.67 & 4 & 1.80 & 2 & \textbf{0.93} \\
						\midrule
						\multirow{4}{*}{\rotatebox{90}{\DatasetWeather}}&\multirow{2}{*}{June}&1 & 11 & 4.25 & 5 & \textbf{2.50} & 6 & \textbf{2.50} \\
						&&2 & 8 & 2.13 & 4 & \textbf{1.73} & 6 & 1.87 \\
						&\multirow{2}{*}{July}&1 & 13 & 4.75 & 4 & \textbf{2.00} & 7 & 2.50 \\
						&&2 & 4 & \textbf{1.20} & 4 & 1.33 & 5 & 2.00 \\
						\bottomrule
		\end{tabular}
	}
	\label{tab:case-study-results-frontend-times-computation}
\end{sidewaystable}

We give an overview over the time measurements in \cref{tab:case-study-results-frontend-times-computation,tab:case-study-results-frontend-times-duration}.

From \cref{tab:case-study-results-frontend-times-computation}, we can deduce that, starting one movement after the other, i.e., performing a naive transition, often requires more computation time than firing multiple movements at once.
However, in spite of the measurable difference in computation times, computing a transition took on average at most twelve milliseconds, which is negligible.

\begin{sidewaystable}
	\centering
	\caption{Maximum ($\max$) and average ($\varnothing$) duration time in milliseconds for each transition style in each setting. Best average duration times per row are highlighted bold.}
	{\begin{tabular}{clcrrrrrr}
						\toprule
						&& & \multicolumn{2}{c}{Naive transitions} & \multicolumn{2}{c}{DAG-based transitions} & \multicolumn{2}{c}{Simultaneous transitions}\\
						\cmidrule{4-9}
						Dataset & Scenario & Setting & \multicolumn{1}{c}{$\max$} & \multicolumn{1}{c}{$\varnothing$} & \multicolumn{1}{c}{$\max$} & \multicolumn{1}{c}{$\varnothing$} & \multicolumn{1}{c}{$\max$} & \multicolumn{1}{c}{$\varnothing$}\\
						\midrule
						\multirow{6}{*}{\rotatebox{90}{\DatasetAccidents}}&\multirow{2}{*}{Bavaria (Munich)}& 1  & 46802 & 23814.50  & 4684 & 3193.25  & 2736 & \textbf{1802.75} \\
						& & 2  & 40949 & 26301.40  & 5006 & 3730.73  & 2525  & \textbf{2077.13} \\
						&\multirow{2}{*}{Bavaria (Nuremberg)}&1  & 2835 & 1174.25 & 1893 & 936.50 & 1863 & \textbf{928.75} \\
						&& 2  & 5496 & 2501.33 & 3745 & 1755.47 & 2810 & \textbf{1694.33} \\
						&\multirow{2}{*}{Berlin}&1   & 5576 & 2383.50  & 2789 & 1458.50  & 2807 & \textbf{1246.00} \\
						& & 2  & 12931 & 7100.27  & 4757 & 3225.47  & 2875 & \textbf{2612.73} \\
						\midrule
						\multirow{8}{*}{\rotatebox{90}{\DatasetTwitter}} &\multirow{2}{*}{Italy}&1  & 990 & 481.75  & 988 & \textbf{478.00}  & 993 & 480.75 \\
						&&2  & 4837 & 2306.80  & 3830 & 1911.00  & 2899 & \textbf{1780.87} \\
						&\multirow{2}{*}{Leeds}&1  & 15561 & 6846.25  & 3989 & 2140.50 & 2472 & \textbf{1420.50} \\
						&&2  & 12091 & 9156.80  & 5081 & 3409.33 & 2400 & \textbf{2212.53} \\
						&\multirow{2}{*}{New Delhi}&1  & 10446 & 3633.75  & 4724 & 1994.75 & 2878 & \textbf{1536.75} \\
						&&2  & 10360 & 5930.00  & 4511 & 2773.27  & 2650 & \textbf{2374.47} \\
						&\multirow{2}{*}{S\~ao Paulo}&1  & 6624 & 2131.50  & 4739 & 1659.00  & 2849 & \textbf{1184.25} \\
						&&2  & 4730 & 2014.47 & 3793 & 1894.40 & 2940 & \textbf{1760.80} \\
						\midrule
						\multirow{4}{*}{\rotatebox{90}{\DatasetWeather}}&\multirow{2}{*}{June}&1  & 10417 & 3913.25  & 3129 & 1699.75 & 2234 & \textbf{1249.00} \\
						&&2  & 11045 & 4247.00 & 5510 & 2452.73 & 2220 & \textbf{1484.53} \\
						&\multirow{2}{*}{July}&1  & 2750 & 1249.25 & 1847 & 1021.75 & 1838 & \textbf{1019.25} \\
						&&2  & 4078 & 1757.60  & 2774 & 1283.13 & 2137 & \textbf{1172.20} \\
						\bottomrule
			
		\end{tabular}
	}
	
	\label{tab:case-study-results-frontend-times-duration}
\end{sidewaystable}

This is in contrast to the duration times for the different transition styles, where the difference is also perceivable.
In particular, \cref{tab:case-study-results-frontend-times-duration} reveals that simultaneous transitions are the fastest, and the naive transitions the slowest.
Noteworthy is that DAG-based transitions are always faster than naive transitions and in particular the duration of the former style is most of the times even close to the simultaneous transitions.
This suggests that performing movements simultaneously that are unrelated in~$G_{\mathcal{M}}$, as we do for DAG-based transitions, has a significant benefit with respect to the duration of a transition.
Note that for settings with less than two movements, all three transitions have the same duration.	Furthermore, for transitions with only movements that are unrelated in~$G_{\mathcal{M}}$, simultaneous and DAG-based transitions take the same time.
Any difference reported in \cref{tab:case-study-results-frontend-times-duration} is due to measurement inaccuracies and negligible.

We now turn our focus to the number of overlaps that we could observe for the different transition styles.
To that end, we considered two labels as overlapping if their respective bounding boxes overlap in at least one percent of their area.
This ensures that we only count overlaps that are clearly perceived as such by a user.

\begin{sidewaystable}
	\centering
	\caption{Maximum ($\max$), average ($\varnothing$), and total ($\Sigma$) number of overlaps for each transition style in each setting.
		We also report on the total number of moving labels ($\Sigma$ \#M) and highlight in bold the transition style with the fewest average overlaps per row.}
	{%
			\begin{tabular}{clcrrrrrrrrrr}
								\toprule
								&& && \multicolumn{3}{c}{Naive transitions} & \multicolumn{3}{c}{DAG-based transitions} & \multicolumn{3}{c}{Simultaneous transitions}\\
								\cmidrule{5-13}
								Dataset & Scenario & Setting & $\sum$ \# M & $\max$ & $\varnothing$ & $\sum$ & $\max$ & $\varnothing$ & $\sum$ & $\max$ & $\varnothing$ & $\sum$\\
								\midrule
								\multirow{6}{*}{\rotatebox{90}{\DatasetAccidents}}&\multirow{2}{*}{Bavaria (Munich)}& 1 & 97 & 7 & 2.25 & 9 & 5 & \textbf{1.75} & 7 & 6 & 2.50 & 10 \\
								&& 2 & 447 & 5 & 2.47 & 37 & 3 & \textbf{0.80} & 12 & 6 & 1.67 & 25 \\
								&\multirow{2}{*}{Bavaria (Nuremberg)}& 1 & 2 & 0 & \textbf{0.00} & 0 &  0 & \textbf{0.00} & 0& 0 & \textbf{0.00} & 0 \\
								&& 2 & 24 & 0 & \textbf{0.00} & 0 & 0 & \textbf{0.00} & 0 & 1 & 0.07 & 1 \\
								&\multirow{2}{*}{Berlin}& 1 & 7 & 1 & 0.25 & 1 & 0 & \textbf{0.00} & 0 & 0 & \textbf{0.00} & 0 \\
								&& 2 & 72 & 2 & 0.60 & 9 & 1 & \textbf{0.27} & 4 & 2 & 0.40 & 6 \\
								\midrule
								\multirow{8}{*}{\rotatebox{90}{\DatasetTwitter}} &\multirow{2}{*}{Italy}& 1 & 0 & 0 & \textbf{0.00} & 0  & 0 & \textbf{0.00} & 0 & 0 & \textbf{0.00} & 0 \\
								&& 2 & 18 & 1 & \textbf{0.13} & 2 & 1 & \textbf{0.13} & 2 & 1 & \textbf{0.13} & 2 \\
								&\multirow{2}{*}{Leeds}& 1 & 27 & 3 & 1.50 & 6 & 3 & \textbf{0.75} & 3 & 3 & 1.00 & 4 \\
								&& 2 & 150 & 5 & 2.27 & 34 & 4 & \textbf{1.07} & 16 & 4 & 1.40 & 21 \\
								&\multirow{2}{*}{New Delhi}& 1 & 8 & 0 & \textbf{0.00} & 0 & 0 & \textbf{0.00} & 0 & 0 & \textbf{0.00} & 0 \\
								&& 1 & 67 & 4 & 0.53 & 8 & 2 & \textbf{0.27} & 4 & 2 & 0.33 & 5 \\
								&\multirow{2}{*}{S\~ao Paulo}& 1 & 5 & 1 & 0.25 & 1 & 0 & \textbf{0.00} & 0 &  2 & 0.50 & 2 \\
								&& 2 & 12 & 2 & 0.20 & 3 & 1 & \textbf{0.13} & 2 & 2 & 0.20 & 3 \\
								\midrule
								\multirow{4}{*}{\rotatebox{90}{\DatasetWeather}}&\multirow{2}{*}{June}& 1 & 16 & 2 & \textbf{0.50} & 2 & 2 & \textbf{0.50} & 2 & 2 & \textbf{0.50} & 2 \\
								&& 2 & 68 & 4 & 1.53 & 23 & 2 & \textbf{0.67} & 10 & 3 & 0.93 & 14 \\
								&\multirow{2}{*}{July}& 1 & 3 & 0 & \textbf{0.00} & 0 & 0 & \textbf{0.00} & 0 & 0 & \textbf{0.00} & 0 \\
								&& 2 & 23 & 1 & \textbf{0.07} & 1 & 1 & \textbf{0.07} & 1 & 1 & 0.20 & 3 \\
								\bottomrule
			\end{tabular}
	}
	\label{tab:case-study-results-frontend-overlaps}
\end{sidewaystable}

From \cref{tab:case-study-results-frontend-overlaps}, we can see that the DAG-based transitions produce in every setting the least number of overlaps, sometimes even half as many overlaps in total as the other styles.
Recall our second interaction setting where increased fifteen times the time of interest by clicking on the $+$ button next to the time slider.
This gives us for every combination of dataset and scenario fifteen transitions that arise from the same interaction pattern.
In \cref{fig:cumulative-overlaps}, we plot for each of these combinations the cumulative number of overlaps for each transition style.
They make immanent that the DAG-based transitions consistently outperforms the other two transition styles in terms of overlaps across the whole interaction span.
Furthermore, \cref{fig:cumulative-overlaps} leads us to the following two observations.
Firstly, although the worst-case upper bound for the number of overlaps for a simultaneous transition is much closer to that for a naive transition, the measured number of overlaps is often closer to that for a DAG-based transition.
Secondly, especially in areas with a dense labeling, DAG-based transitions tends to produce significantly fewer overlaps; see for example Munich-setting.

\begin{figure}
	\centering
	\includegraphics[width=\linewidth]{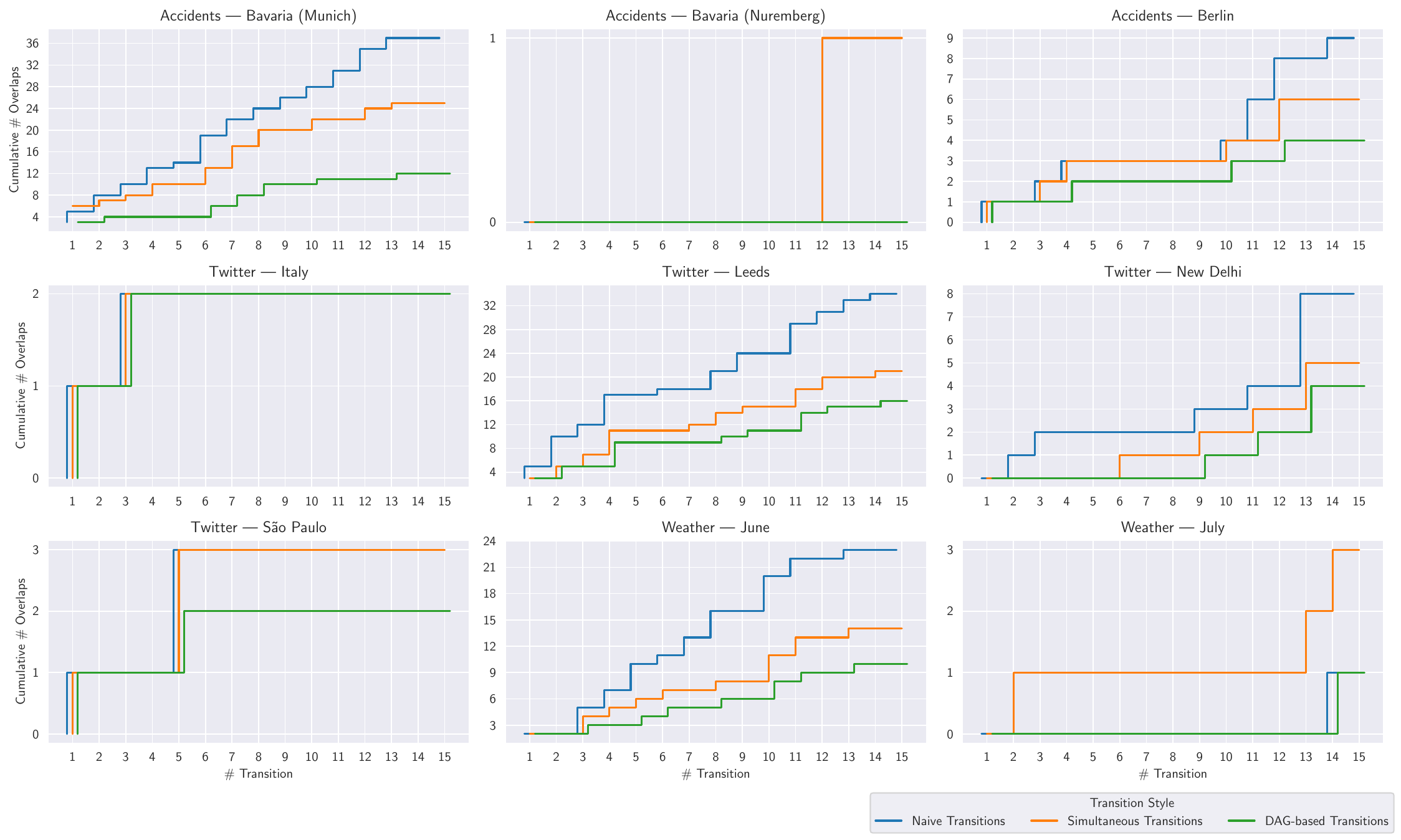}
	\caption{The measured cumulative number of overlaps in the second interaction setting for each combination of dataset and scenario and each transition style.}
	\label{fig:cumulative-overlaps}
\end{figure}

\subsection{Discussion and Limitations}
\label{sec:subsec:case-study-discussion}
Our case study shows that DAG-based transitions find a good compromise between the number of overlaps (\GoalOverlapas; up to half as many overlaps compared to other transition styles) and the duration of the transitions (\GoalDuration; on average around 40\% faster than naive transitions).
However, simultaneous transitions are more appealing, if one favours faster transitions over the number of overlaps (on average around 19\% faster than DAG-based transitions).
Hence, we see a clear trade-off between the number of overlaps and transition duration, similarly to previous work~\cite{BergG13}.
The case study also indicates that DAG-based or simultaneous transitions should be preferred over naive transitions in terms of transition duration and number of overlaps.
This underlines the role of naive transitions primarily serving as a simple baseline approach on which some of the more advanced transition styles such as DAG-based transitions are built and against which they can be compared.
While the case study underlines the practical merit of transitions, we must also acknowledge some limitations.
First, we only use a naive, greedy algorithm to compute the individual labelings.
As our case study focuses on the transition between two arbitrary labelings, this fact can be neglected in our setting.
However, we must acknowledge that the individual labelings have a great effect on the readability and quality of the map.
Furthermore, the perceived stability of the labeling also depends on the computation time of the new labelings and the number of label changes, two factors that are beyond the scope of this case study.
Finally, while this case study confirmed that the proposed transition styles meet their theoretical expectations in terms of the transition goals, i.e., duration and number of label overlaps, it should be verified in a user study if these are the right criteria to optimize for.

\section{Conclusion}
\label{sec:conclusion}
In this paper we performed a first investigation into the number of overlaps produced by transitions on labelings of points, and started by proving tight upper bounds for various transition styles.
In addition, we implemented the transition styles in a prototype and performed a case study that revealed the need for sophisticated transition styles that find a good compromise between the number of overlaps and the duration of a transition.
We see this paper as a first step towards understanding such transitions in point labeling.
We have already addressed in the previous section some directions for future work targeted to our case study.
Also in terms of theoretical questions, there are numerous open questions for future work, such as:
\begin{itemize}
	\item Should we develop new transition styles or improve the existing ones? Can we utilize more structured movement, like performing all movements in the same direction simultaneously?
	Can we develop map feature-aware transition styles that, for example, avoid occluding important map features during movement?
	\item Is it sensible to try to formalize more perception-oriented desiderata for transitions, such as the symmetry of transitions or the traceability of labels?
	\item Is choosing label directions in simultaneous transitions still \NP-hard with unit weights?
	\item Can we compute a \emph{stable} labeling \Labeling{2}, that minimizes the number of moving labels?
	\item How do transitions perform in comparison that do not separate additions, movements, and removals with respect to our optimization goals but also in cognitive studies assessing the usability and utility of the different approaches?
\end{itemize}

\bibliography{references}

\newpage
\appendix
\section{Arbitrary Rectangle Labels}\label{app:arbitrary-rectangles}
While in the main text we considered only square labels, point labelings often use arbitrary rectangles. If we allow our labels to be arbitrary rectangles, then it is no longer guaranteed that only one (stationary) label can overlap with the area traversed by the moving label.
If we assume that the label with the largest side width (some $\sigma_{x_{\max}}$) must perform a diagonal movement, we can align $\frac{\sigma_{x_{\max}}}{\sigma_{x_{\min}}}$ stationary labels with a width of $\sigma_{x_{\min}}$, the smallest label width in our map, on the horizontal edge of the traversed area. 
As one label can always extend out of that traversed area without resulting in an invalid labeling \Labeling{1} or \Labeling{2}, we can put up to $\lceil \frac{\sigma_{x_{\max}}}{\sigma_{x_{\min}}} \rceil$ labels next to each other in the $x$-direction.
The same holds for the $y$-axis, with maximum and minimum height $\sigma_{y_{\max}}$ and $\sigma_{y_{\min}}$, respectively.
As we can put labels anywhere inside the traverse area, we can place up to $\lceil \frac{\sigma_{x_{\max}}}{\sigma_{x_{\min}}} \rceil \cdot \lceil \frac{\sigma_{y_{\max}}}{\sigma_{y_{\min}}} \rceil$ labels intersecting that area and therefore $\lceil \frac{\sigma_{x_{\max}}}{\sigma_{x_{\min}}} \rceil \cdot \lceil \frac{\sigma_{y_{\max}}}{\sigma_{y_{\min}}} \rceil$ overlaps occur during the movement (see \cref{fig:corollay-arbitrary-rectangles}).
This results in the following corollary, as an extension of Lemma 2.1.

\begin{figure}
	\centering
	\includegraphics{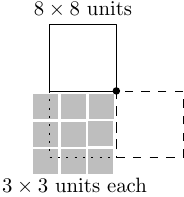}
	\caption{The $8 \times 8$ label wants to perform a counterclockwise diagonal movement and overlaps with nine stationary $3 \times 3$ labels.
		The overlapping region is dotted, while the end position of the moving label is indicated with the dashed rectangle.}
	\label{fig:corollay-arbitrary-rectangles}
\end{figure}

\begin{corollary}
	\label{cor:lem-1-arbitrary-rectangles}
	When the labels are arbitrary rectangles with side length $\sigma_{x_i}$ and $\sigma_{y_i}$, with $1 \leq i \leq n$ and $n$ denotes the number of labels, performing transition \Transition{R M_i A} for a label of a point $p_i$ can result in at most ${\lceil \frac{\sigma_{x_{\max}}}{\sigma_{x_{\min}}} \rceil \cdot \lceil \frac{\sigma_{y_{\max}}}{\sigma_{y_{\min}}} \rceil}$ overlaps given that the end position of $p_i$ is free, where $\sigma_{x_{\min}} = \min\{\sigma_{x_i}\mid 1 \leq i \leq n\}$ and $\sigma_{x_{\max}}$, $\sigma_{y_{\min}}$ and $\sigma_{y_{\max}}$ are defined similarly.
\end{corollary}

\cref{cor:lem-1-arbitrary-rectangles} shows how upper bounds on the number of overlaps produced by square labels can be extended to the setting of arbitrary rectangles. This introduces only a constant factor, depending on the ratio between the largest and smallest side lengths in each dimension. However, for many transitions adding the constant factor, as suggested by \cref{cor:lem-1-arbitrary-rectangles}, does not yield a tight bound. This stems from the fact that many upper bounds require overlaps with the start or end position of a label $l$, not just the traversed area of $l$. Since those positions are solely occupied at respectively the beginning and the end of the transition, we cannot place $\lceil \frac{\sigma_{x_{\max}}}{\sigma_{x_{\min}}} \rceil \cdot \lceil \frac{\sigma_{y_{\max}}}{\sigma_{y_{\min}}} \rceil$ labels in those positions: many of those labels are unable to move away completely.

\end{document}